\documentclass[a4paper,11pt]{article}

\usepackage[english]{babel}
\usepackage{amsmath}
\usepackage{amsthm}
\usepackage{enumerate}
\usepackage{graphicx}
\usepackage{mathrsfs}
\usepackage{amssymb}
\usepackage[usenames,dvipsnames]{color}
\usepackage{empheq}
\usepackage{hyperref}
\hypersetup{
colorlinks = true, 
linkcolor = {blue}
}
\usepackage[affil-it]{authblk}

\newtheorem{prop}{Proposition}
\newtheorem{lemma}[prop]{Lemma}
\newtheorem{thm}[prop]{Theorem}
\newtheorem{oss}[prop]{Remark}

\theoremstyle{definition}
\newtheorem{defn}[prop]{Definition}
\newtheorem{cor}[prop]{Corollary}

\begin{document}

\title{Gap probabilities for the Generalized Bessel process: a Riemann-Hilbert approach}
\author{Manuela Girotti
\thanks {email: \texttt{mgirotti@mathstat.concordia.ca}}}
\affil{\small Department of Mathematics and Statistics, Concordia University \\
1455 de Maisonneuve Ouest, Montr\'eal, Qu\'ebec, Canada, H3G 1M8}
\date{}

\maketitle

\tableofcontents

\begin{abstract}
We consider the gap probability for the Generalized Bessel process in the single-time and multi-time case. We prove that the scalar and matrix Fredholm determinants of such process can be expressed in terms of  determinants of Its-Izergin-Korepin-Slavnov integrable kernels and thus related to suitable Riemann-Hilbert problems. In the single-time case, we construct a Lax pair formalism, while in the multi-time case we explicitly define a new multi-time kernel to study.
\end{abstract}

\section{Introduction}
%\begin{itemize}
%\item bla bla - same thing as in the Bessel, but shorter: IITS integrable kernel, Fredholm determinant and $\tau$-function.
%\item Cite Arno's articles \cite{Kuij} which is the origin of the kernel and Balint's work \cite{Balint}.
%\item Spend a few words justifying the name you gave to the kernel: same contour setting, higher order Lax pair as Bessel. 
%\item Specify that this is a completely new result, since nobody worked on this before and also the multi-time formulation is brand new. 
% \item Our multi-time and Balint's multi-time (autonomously derived) are equivalent up to transposition and translation of the parameter $\tau$. 
%\item Write some "intro" at the beginning of every section/subsection, just to clarify what are you gonna talk about there
%\item especially along the single-time section, point out the similarities with the Bessel! (same contours, same exponential plus a term $1/z^2$, same phase $\theta$ plus a term $1/\lambda^2$, same Lax pair but higher Poincar\'e rank at zero)
%\end{itemize}

The Generalized Bessel process is a determinantal point  process \cite{Soshnikov} defined in terms of a trace-class integral operator acting on $L^2(\mathbb{R}_+)$, with kernel
\begin{equation}
K^{GEN} (x,y;\tau) = \int_{\hat \gamma} \frac{ds}{2\pi i }\int_{\gamma} \frac{dt}{2\pi i} \, \frac{e^{-xs - \frac{\tau}{s} + \frac{1}{2s^2} + yt + \frac{\tau}{t} - \frac{1}{2t^2}}}{s-t} \left(\frac{s}{t} \right)^\nu \label{kernelcrit2}
\end{equation}
with $\nu>-1$; the logarithmic cut is on $\mathbb{R}_-$. The curve $\gamma$ and $\hat \gamma$ are described in Figure \ref{ArnoGenBes}.

The Generalized Bessel kernel was first introduced as a critical kernel by Kuijlaars \textit{et al.} in \cite{Kuij2} and \cite{Kuij}. In these articles, a model of non-intersecting squared-Bessel paths, starting at time $t = 0$ at the same positive value $x=\kappa>0$ and ending at time $t = 1$ at $x = 0$, was proposed and studied.

The positions of the paths at any given time $t\in (0,1)$  are a determinantal point process with correlation kernel built out of the transition probability density of the squared Bessel process. In \cite{Kuij2}, it was proven that, after appropriate scaling, the paths fill out a region in the $tx$-plane as in \figurename \ \ref{arnoSQB}: the paths stay initially away from the  axis $x = 0$, but at a certain critical time $t^*$ the smallest paths come to the hard edge $x=0$ and then remain close to it. 

\begin{figure}[!h]
\centering
\includegraphics[width=.8\textwidth]{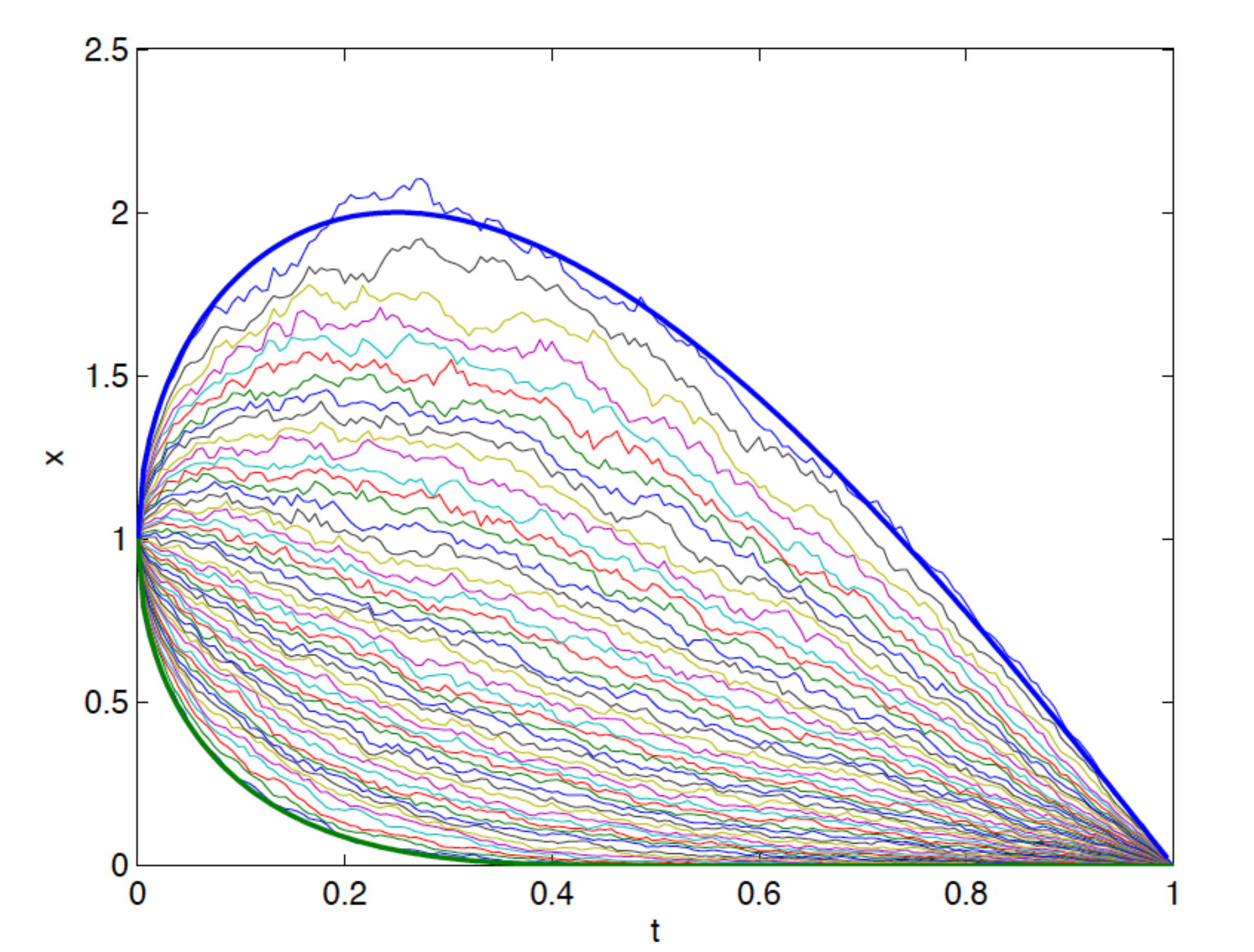}
\caption{Numerical simulation of 50 rescaled non-intersecting Squared Bessel Paths with $\kappa = 1$ (taken from \cite{Kuij2}).}
\label{arnoSQB}
\end{figure}

As the number of paths tends to infinity, the local scaling limits of the correlation kernel are the following: the sine kernel appears in the bulk, the Airy kernel at the soft edges, i.e. the upper boundary for all $t\in (0,1)$ and the lower boundary of the limiting domain for $t<t^*$, while for $t>t^*$, the Bessel kernel appears at the hard edge $x=0$, see \cite[Theorems 2.7-2.9]{Kuij2}. 

Thus, in the critical time $t=t^*$, there is a transition between the Airy and the Bessel kernel and the dynamics at that point is described by the new (critical) kernel (\ref{kernelcrit2}).

We prefer to refer to this kernel, and to its multi-time counterpart, as  ``Generalized Bessel kernel" because of the several analogies with the Bessel kernel (\cite{Me}) appearing along our study. First of all, the contours setting is the same one as for the Bessel kernel (see \cite{Me}); many of the calculations performed in \cite{Me} for the Bessel kernel are here reproduced with very few adjustments. Moreover, as it will be clear in Section \ref{1timeGENB}, gap probabilities of the Generalized Bessel operator are related to a Lax pair which belongs to a higher order Painlev\'e III hierarchy, while it is known that the gap probabilities of the Bessel kernel is related to the Painlev\'e III transcendent, as shown in \cite{Me}.

%Dyson, in [1], described how to implement a dynamics into random matrix models in such a way that the eigenvalues of the matrix behave like finitely many non-intersecting Brownian motions on the real line. 

%In a suitable scaling regime we are lead to the study of certain interesting time-dependent determinantal random point processes generalizing the typical probability distributions appearing in random matrix models (see [2]). 

%
%The interpretation of the process is that of a ÔÔgasÕÕ of infinitely many particles undergoing diffusion according to a squared Bessel process ($\text{BESQ}^2$) and whose paths are conditioned to start at time $t = 0$ at the positive value $x=\kappa>0$ and end at time $t = 1$ at $x = 0$, not  intersecting at any time. 

In this paper we will show that the  correlation functions related to the Generalized Bessel process can be expressed in terms of determinants of a kernel $K(x, y)$ (matrix-valued in the multi-time case) in the sense of Its-Izergin-Korepin-Slavnov (\cite{IIKS}). Namely, the kernel can be written in the form
\begin{gather}
K(x,y) = \frac{\textbf{f}^T(x) \cdot \textbf{ g} (y)}{x-y}
\end{gather}
where $\textbf{ f}, \textbf{ g}$ are two matrices of a given dimension such that $\textbf{ f}^T(x) \cdot \textbf{ g}(x) =0$.
We refer to the review paper \cite{Soshnikov} that also explains how the ``gap probabilities" (probability of having no particles in certain regions) are related to Fredholm determinants.

%One important feature of the probability distributions appearing in random matrices is that they are related to Fredholm determinants of integrable operators in the sense of Its-Izergin-Korepin-Slavnov (\cite{IIKS}). Namely, their kernels can be written in the form
%\begin{gather}
%K(x,y) = \frac{\vec f^T(x) \vec g (y)}{x-y}
%\end{gather}
%where $\vec f, \vec g$ are two matrices of a given dimension such that $\vec f^T(x) \vec g(x) =0$.

In such case, the computation of the Fredholm determinant is reduced to the study of a certain Riemann-Hilbert problem canonically related to the operator's kernel (for a concise account, see \cite{JohnSasha}). The Riemann-Hilbert formulation is very useful for finding some differential equations and studying asymptotic properties of the determinant. Moreover, it will be possible to connect it to the Jimbo-Miwa-Ueno $\tau$ function.

%In this paper, we will prove that the determinants of the Generalized Bessel operator in single-time and multi-time regime are equal to the determinants of some explicitly given IIKS integrable kernels. 
%%The main steps in our study of the gap probabilities are the following:  we will first find an integrable operator, acting on $L^2(\Sigma)$, with $\Sigma$ a suitable collection of contours. Through an appropriate Fourier transform, we will prove that such operator has the same Fredholm determinant as the Bessel process. 
%We will then set up a Riemman-Hilbert problem for such integrable kernel and connect it to the Jimbo-Miwa-Ueno $\tau$ function. 

%This strategy will be applied separately to both the single-time and the multi-time Bessel process. Our approach derives from the one used in \cite{MeMmulti} and \cite{MeM} for the Airy and Pearcey processes in the dynamic and time-less regime respectively.

%Whereas the part dedicated to the single-time process is mostly a review of known results (see \cite{Ptrans}, \cite{JMUII} and \cite{TWBessel}), re-derived using an alternative approach, the results on the multi-time Bessel are new and never appeared in the literature before.

Our approach is the same as the one used in \cite{MeM} for the scalar Airy and Pearcey operators, in \cite{MeMmulti} for the matrix Airy and Pearcey operators and in \cite{Me} for the Bessel operator. 

As an example of possible applications we describe how to obtain a system of isomonodromic Lax equations for the (single-time) process. 

Moreover, having a RiemannÐHilbert formulation for these Fredholm determinants will allow the study of asymptotics of Generalized Bessel gap probabilities and their connection with Airy and Bessel gap probabilities, using steepest descent methods, along the lines of \cite{MeM}. 

The formulation of the multi-time Generalized Bessel kernel is a completely new result and its derivation has been addressed in the Appendix. An equivalent formulation has been proposed and autonomously derived by Delvaux and Veto (\cite{Balint}).

The paper is organized as follows: in section \ref{1timeGENB} we will deal with the single-time Generalized Bessel operator restricted to a generic collection of intervals; in the subsection \ref{1timeGENB0a} we will focus on the single-time Generalized Bessel process restricted to a single interval $[0,a]$: we will find a Lax pair and we will be able to make a connection between the Fredholm determinant and a Painlev\'e III hierarchy. 
%This provides a different and direct proof of this known connection (\cite{JMUII}, \cite{TWBessel}); in particular our approach directly specifies the monodromy data of the associated isomonodromic system and allows to use the steepest descent method to investigate asymptotic properties, if so desired. 
In section \ref{ntimeGENB} we will study the gap probabilities for the multi-time Bessel process. %Although the results of section \ref{ntimeGENB} strictly include those of  section \ref{1timeGENB}, we have decided to separate the two cases for the benefit of a clearer exposition. 
In the Appendix, we show how we found the multi-time Generalized Bessel kernel and we make a comparison with Delvaux and Veto's one. We prove that these two kernels are equivalent up to a transposition of the operator and a translation of the parameter $\tau$.

\section{Single-time Generalized Bessel}\label{1timeGENB}

The Generalized Bessel kernel is
\begin{gather}
K^{GEN}(x,y;\tau) = 
%\int_{\Gamma} \frac{ds}{2\pi i }\int_{\Sigma} \frac{dt}{2\pi i} \, \frac{e^{xs + \frac{\tau}{s} + \frac{1}{2s^2} - yt - \frac{\tau}{t} - \frac{1}{2t^2}}}{t-s} \left(\frac{s}{t} \right)^\alpha \nonumber \\
  \int_{\gamma \times \hat \gamma} \frac{dt \, ds}{(2\pi i)^2 } \, \frac{e^{\phi_\tau(y,t) - \phi_\tau(x,s) }}{s-t} \left(\frac{s}{t} \right)^\nu \\
  \phi_\tau(z,t) := zt+ \frac{\tau}{t}-\frac{1}{2t^2} \label{phaseGEN}
\end{gather}

\begin{figure}[!h]
\centering
\includegraphics[width =.7\textwidth]{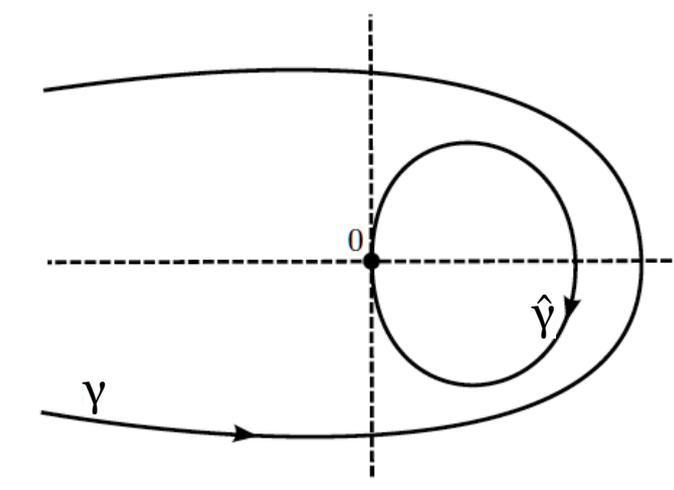}
\caption{The curves appearing in the definition of the Generalized Bessel kernel.}
\label{ArnoGenBes}
\end{figure}

%\begin{equation}
%K^{GEN}_\nu(x,y;\tau) = \int_{\Gamma} \frac{ds}{2\pi i }\int_{\Sigma} \frac{dt}{2\pi i} \, \frac{e^{xs + \frac{\tau}{s} + \frac{1}{2s^2} - yt - \frac{\tau}{t} - \frac{1}{2t^2}}}{t-s} \left(\frac{s}{t} \right)^\nu
%\end{equation}
where $\tau \in \mathbb{R}$ is a fixed parameter, the contour $\hat \gamma$ is a closed loop in the right half-plane tangent to the origin and oriented clockwise, while the contour $\gamma$ is an unbounded loop oriented counterclockwise and encircling $\hat \gamma$; the logarithmic cut lies on $\mathbb{R}_-$ (see \figurename \ \ref{ArnoGenBes}).
%For more history, see \cite{Kuij}. 

\begin{oss}
The curve setting is equivalent to the curve setting appearing in the definition of the Bessel kernel (see \cite{Me}). 

Moreover, the phase appearing in the exponential (\ref{phaseGEN}) resembles the Bessel kernel one $\psi(z,t):= zt -\frac{1}{4t}$ with an extra term which introduces a higher singularity at $0$.
\end{oss}

Our interest is focused on the gap probability of such operator restricted to a collection of intervals $I$, i.e. the quantity
\begin{gather}
\det (\operatorname{Id} - K^{GEN}\chi_I)
\end{gather}
with $\chi_I$ the characteristic function of the Borel set $I$.
 
\begin{oss}
 Let's consider a multi-interval $I:= \bigcup_{k=1}^N [a_{2k-1},a_{2k}]$. 
 %\cup [a_3, a_4] \cup \ldots \cup [a_{N-1}, a_{N}]$, $N$ even.  
 Given $K^{GEN}_a :=  \chi_{[0,a]}(x) K^{GEN}(x,y)$, then we have
\begin{gather}
\chi_I (x) \cdot K^{GEN} (x,y) 
%= K^{GEN}_{a_2} - K^{GEN}_{a_1} + K^{GEN}_{a_4} - K^{GEN}_{a_3} + \ldots + K^{GEN}_{a_N}-K^{GEN}_{a_{N-1}} \nonumber \\
= \sum_{j=1}^{2N} (-1)^j K^{GEN}_{a_j}(x,y)
\end{gather}
\end{oss}

\begin{oss} The Generalized Bessel operator is not trace class at infinity; indeed, the kernel $K^{GEN}(x,x;\tau)$ is not integrable in a neighbourhood of $\infty$.
\end{oss}

The first step in our study is to establish a relation between the Generalized Bessel operator and a suitable integrable operator in the sense of Its-Izergin-Korepin-Slavnov (IIKS; see \cite{IIKS}).

\begin{thm}\label{identityFD}
Given a collection of (disjoint) intervals $I:= \bigcup_{k=1}^N [a_{2k-1},a_{2k}]$, the following identity between Fredholm determinats holds
\begin{equation}
\det \left( \operatorname{Id} - \chi_I K^{GEN} \right) = \det \left( \operatorname{Id} - \mathbb{K}^{GEN} \right)
\end{equation}
where $\mathbb{K}^{GEN}$ is an IIKS integrable operator acting on $L^2(\gamma \cup \hat \gamma)$ with kernel
\begin{align}
&\mathbb{K}^{GEN} (t,s) 
%= \frac{\sum_{j} (-1)^j e^{t \left( a_j- \frac{a_1}{2}\right) -a_js - \frac{\tau}{s} + \frac{1}{2s^2} } s^\nu \cdot \chi_{\gamma}(t) \chi_{\hat \gamma}(s) + e^{ \frac{sa_1}{2} + \frac{\tau}{s} - \frac{1}{2s^2} } s^{-\nu}\cdot \chi_{\gamma}(s) \chi_{\hat \gamma}(t)}{2\pi i (t-s)} \nonumber \\
= \frac{\vec{f}\,^T (t) \cdot \vec{g}(s)}{t-s} \label{Kcritfg} \\
&\vec{f}(t) = \frac{1}{2\pi i } \left[\begin{array}{c}1 \\ 0 \\ \vdots \\ \vdots \\ 0  \end{array}\right] \chi_{\hat \gamma}(t) +  \frac{1}{2\pi i } \left[\begin{array}{c} 0 \\ e^{ \frac{ta_1}{2}} \\ e^{t\left(a_2 - \frac{a_1}{2} \right)} \\ \vdots \\ e^{t\left(a_{2N} - \frac{a_1}{2} \right)}  \end{array}\right] \chi_{\gamma}(t)  \label{Kcritf}\\
&\vec{g}(s) = \left[\begin{array}{c} 0 \\ - e^{ -a_1s - \frac{\tau}{s} + \frac{1}{2s^2} } s^\nu \\ e^{-a_2s - \frac{\tau}{s} + \frac{1}{2s^2} } s^\nu \\ \vdots \\ (-1)^{2N} e^{ -a_{2N}s - \frac{\tau}{s} + \frac{1}{2s^2} } s^\nu  \end{array}\right]\chi_{\hat \gamma}(s) + \left[ \begin{array}{c}  e^{ \frac{sa_1}{2} + \frac{\tau}{s} - \frac{1}{2s^2} } s^{-\nu} \\ 0 \\ \vdots \\ \vdots \\ 0 \end{array}\right] \chi_{\gamma} (s) \label{Kcritg}
\end{align}
\end{thm}

\begin{proof}
Since the preliminary calculations are linear, we will start working on the single term $K^{GEN}_{a_j}$ and we will later sum them up over the $j = 1, \ldots, 2N$.
% We notice again that the FT is linear, so all that has already been written for the single interval case can be easily extended to the multi-interval kernel $K^{GEN} = \sum (-1)^j K^{GEN}_{j}$. 

%{\color{red} Again, it looks like the interval $(-\infty, a]$ is equivalent to $[0,a]$, since for $y<0$ we get a zero residue (same argument as in Bessel). What about $x<0$?}

\begin{gather}
K^{GEN}_{a_j} := \chi_{[0,a_j]}(x) K^{GEN}(x,y;\tau) \nonumber \\
%\chi_{(-\infty,a)}(x) K^{GEN}(x,y;\tau) = \nonumber \\
= \int_{i\mathbb{R} + \epsilon} \frac{d\xi}{2\pi i} \frac{e^{\xi(a_j-x)}}{\xi-s} \int_{\gamma \times \hat \gamma} \frac{dt \, ds}{(2\pi i)^2 } \, \frac{e^{-a_j s - \frac{\tau}{s} + \frac{1}{2s^2} + yt + \frac{\tau}{t} - \frac{1}{2t^2}}}{s-t} \left(\frac{s}{t} \right)^\nu \nonumber \\
= \int_{i\mathbb{R} + \epsilon} \frac{d\xi}{2\pi i} e^{-x \xi } \int_{i\mathbb{R} + \epsilon} \frac{dt}{2\pi i} e^{yt}  \int_{\hat \gamma} \frac{ds}{2\pi i } \, \frac{e^{\xi a_j + \frac{\tau}{t} - \frac{1}{2t^2} -a_js - \frac{\tau}{s} + \frac{1}{2s^2} }}{(\xi -s)(s-t)} \left(\frac{s}{t} \right)^\nu 
\end{gather}
where we continuously deformed the contour $\gamma$ into a suitably translated imaginary axis. 

Introducing the following Fourier transform operators
\begin{equation}
\renewcommand\arraystretch{2}
\begin{array}{c|c}
\mathcal{F}: L^2(\mathbb{R}) \rightarrow L^2(i\mathbb{R}+\epsilon) & \mathcal{F}^{-1}: L^2(i\mathbb{R}+\epsilon) \rightarrow L^2(\mathbb{R}) \\
f(x) \mapsto \frac{1}{\sqrt{2\pi i}} \int_\mathbb{R}  f(x) e^{\xi x} dx & h(\xi) \mapsto \frac{1}{\sqrt{2\pi i}} \int_{i\mathbb{R}+\epsilon} h(\xi) e^{-\xi x}d\xi
\end{array}\label{Fourier}
\end{equation} 
we can claim that 
\begin{equation}
K^{GEN} = \mathcal{F}^{-1} \circ \mathcal{K}^{GEN} \circ \mathcal{F},
\end{equation} 
$\mathcal{K}^{GEN}:= \sum_j (-1)^j \mathcal{K}^{GEN}_{a_j}$ being an operator acting on $L^2(i\mathbb{R}+\epsilon)$ with kernels
\begin{equation}
\mathcal{K}^{GEN}_{a_j} (\xi, t;\tau) = \int_{\hat \gamma} \frac{ds}{2\pi i } \, \frac{e^{\xi a_j + \frac{\tau}{t} - \frac{1}{2t^2} -a_j s - \frac{\tau}{s} + \frac{1}{2s^2} }}{(\xi -s)(s-t)} \left(\frac{s}{t} \right)^\nu
\end{equation}
$\forall \, j=1, \ldots, 2N$, $\xi, t \in i\mathbb{R} + \epsilon$.

%The expression for each $K^{GEN}_j$ is known form the previous sections:
%\begin{equation}
%K^{GEN}_j(x,y;\tau) =  \int_{\hat \gamma} \frac{ds}{2\pi i }\int_{\gamma} \frac{dt}{2\pi i} \, \frac{e^{-xs - \frac{\tau}{s} + \frac{1}{2s^2} + yt + \frac{\tau}{t} - \frac{1}{2t^2}}}{s-t} \left(\frac{s}{t} \right)^\alpha 
%\end{equation}
%and its FT as well:
%\begin{equation}
%\mathcal{K}^{GEN}_j (\xi, t;\tau) = \int_{\hat \gamma} \frac{ds}{2\pi i } \, \frac{e^{\xi a_j + \frac{\tau}{t} - \frac{1}{2t^2} -a_js - \frac{\tau}{s} + \frac{1}{2s^2} }}{(\xi -s)(s-t)} \left(\frac{s}{t} \right)^\nu
%\end{equation}
%$\forall j= 1, \ldots, N$, $\xi, t \in i\mathbb{R}+\epsilon \sim \gamma$.

In order to ensure the convergence of the kernel, we conjugate it with the function $f(z):= e^{\frac{a_1 z}{2}}$ and, with abuse of notation, we call the resulting kernels $\mathcal{K}^{GEN}$ and $\mathcal{K}^{GEN}_{a_j}$ as well.
\begin{gather}
e^{\frac{a_1t}{2} -\frac{a_1\xi}{2}} \mathcal{K}^{GEN} (\xi, t;\tau)=\sum_{j=1}^{2N} (-1)^j \int_{\hat \gamma} \frac{ds}{2\pi i } \, \frac{e^{\xi \left( a_j - \frac{ a_1}{2}\right)+\frac{t a_1}{2} + \frac{\tau}{t} - \frac{1}{2t^2} -a_js - \frac{\tau}{s} + \frac{1}{2s^2} }}{(\xi -s)(s-t)} \left(\frac{s}{t} \right)^\nu 
%\nonumber \\
%= \sum_{j=1}^N (-1)^j \mathcal{K}^{GEN}_{j,1} = \sum_{j=1}^N \mathcal{B}_{1} \circ \mathcal{A}_{j,1}
\end{gather}
\begin{oss}
We recall that Fredholm determinants are invariant under conjugation by bounded invertible operators.
\end{oss}
We continuously deform the translated imaginary axis $i\mathbb{R}+\epsilon$ into its original shape $\gamma$; note that $a_j - \frac{a_1}{2} > 0$, $\forall j=1,\ldots, 2N$.  It can be easily shown that the operator $\mathcal{K}^{GEN}_{a_j}$ is the composition of two operators for every $j=1, \ldots, 2N$; moreover, it is trace-class.

\begin{lemma}\label{lemma5}
The operators $\mathcal{K}^{GEN}_{a_j}$ are trace-class operators, $\forall \, j=1,\ldots, 2N$, and the following decomposition holds $\mathcal{K}^{GEN}_{a_j} = \mathcal{B}_1 \circ \mathcal{A}_{j,1}$, with 
\begin{equation}
\renewcommand\arraystretch{2}
\begin{array}{c|c}
\mathcal{A}_{j,1}:  L^2(\gamma) \rightarrow  L^2(\hat \gamma) & \mathcal{B}_{1}: L^2(\hat \gamma) \rightarrow L^2(\gamma) \\
h(t) \mapsto  s^\nu e^{-a_js -\frac{\tau}{s} + \frac{1}{2s^2}}  \int_{\gamma} \frac{e^{t\left( a_j - \frac{a_1}{2}\right)}}{t-s} h(t) \, \frac{dt}{2\pi i} & f(s) \mapsto t^{-\nu}e^{\frac{ta_1}{2} + \frac{\tau}{t} - \frac{1}{2\tau^2}} \int_{\hat \gamma} \frac{f(s)}{s-t}   \, \frac{ds}{2\pi i}
\end{array}
\end{equation}
$\mathcal{A}_{j,1}$ and $\mathcal{B}_1$ are trace-class operators themselves.
\end{lemma}

%\begin{gather}
%L^2(\gamma) \xrightarrow{\mathcal{A}_{j,1}} L^2(\hat \gamma) \xrightarrow{\mathcal{B}_1} L^2(\gamma)
%\end{gather}
%with
%\begin{gather}
%\mathcal{A}_{j,1}(\xi, s; \tau)=  \frac{e^{\xi \left( a_j- \frac{a_1}{2}\right) -a_js - \frac{\tau}{s} + \frac{1}{2s^2} }}{2\pi i(\xi -s)} s^\nu \\
%\mathcal{B}_1(s,t;\tau) =  \frac{e^{ \frac{ta_1}{2} + \frac{\tau}{t} - \frac{1}{2t^2} }}{2\pi i(s-t)} t^{-\nu}
%\end{gather}
%I refer to the previous section regarding the proof that $\mathcal{K}^{GEN}$, $\mathcal{A}_{j,1}$ and $\mathcal{B}_1$ are trace-class. 

\begin{proof}
We introduce an additional translated imaginary axis $i\mathbb{R}+\delta$ ($\delta >0$), not intersecting with $\gamma$ and $\hat \gamma$, and we decompose $\mathcal{A}_{j,1}$ and $\mathcal{B}_1$ in the following way: $\mathcal{A}_{j,1} = \mathcal{O}_{j,2} \circ \mathcal{O}_{j,1}$ and $\mathcal{B}_1 = \mathcal{P}_2 \circ \mathcal{P}_1$ with
\begin{equation}
\renewcommand\arraystretch{2}
\begin{array}{c|c}
\mathcal{O}_{j,1}: L^2(\gamma) \rightarrow L^2(i\mathbb{R}+\delta) & \mathcal{O}_{j,2} :L^2(i\mathbb{R}+\delta) \rightarrow L^2(\hat \gamma) \\
\displaystyle f(\xi) \mapsto \int_{\gamma} \frac{d\xi}{2\pi i} e^{\xi\left( a_j - \frac{ a_1}{2}\right)} \frac{f(\xi)}{\xi -w}  & \displaystyle g(w) \mapsto s^\nu e^{-a_js - \frac{\tau}{s} + \frac{1}{2s^2}} \int_{i\mathbb{R} + \delta} \frac{dw}{2\pi i} \frac{g(w)}{w-s} 
\end{array}
\end{equation}
%\begin{gather}
%L^2(\gamma) \xrightarrow{\mathcal{O}_1} L^2(i\mathbb{R}+\delta) \xrightarrow{\mathcal{O}_2} L^2(\hat \gamma) \\
%\xi \rightarrow w \rightarrow s
%\end{gather}
%with
%\begin{gather}
%\mathcal{O}_1 [f(\xi)]  (w) = \int_{\gamma} \frac{d\xi}{2\pi i} e^{\frac{\xi a}{2}} \frac{f(\xi)}{\xi -w}  \\
%\mathcal{O}_2[g(w)](s)= s^\nu e^{-as - \frac{\tau}{s} + \frac{1}{2s^2}} \int_{i\mathbb{R} + \delta} \frac{dw}{2\pi i} \frac{g(w)}{w-s}
%\end{gather}
and
\begin{equation}
\renewcommand\arraystretch{2}
\begin{array}{c|c}
\mathcal{P}_1:L^2(\hat \gamma) \rightarrow L^2(i\mathbb{R}+\delta) & \mathcal{P}_2: L^2(i\mathbb{R}+\delta) \rightarrow L^2(\gamma) \\
\displaystyle f(s)\mapsto \int_{\hat \gamma} \frac{ds}{2\pi i} \frac{f(s)}{s-u}  & \displaystyle g(u) \mapsto  t^{-\nu} e^{\frac{ta_1}{2} + \frac{\tau}{t} - \frac{1}{2t^2}} \int_{i\mathbb{R} + \delta} \frac{du}{2\pi i} \frac{g(u)}{u-t}
\end{array}
\end{equation}
%\begin{gather}
%L^2(\hat \gamma) \xrightarrow{\mathcal{P}_1} L^2(i\mathbb{R}+\delta) \xrightarrow{\mathcal{P}_2} L^2(\gamma) \\
%s \rightarrow u \rightarrow t
%\end{gather}
%with
%\begin{gather}
%\mathcal{P}_1 [f(s)]  (u) = \int_{\hat \gamma} \frac{ds}{2\pi i} \frac{f(s)}{s-u}  \\
%\mathcal{P}_2[g(u)](t)=  t^{-\nu} e^{\frac{ta}{2} + \frac{\tau}{t} - \frac{1}{2t^2}} \int_{i\mathbb{R} + \delta} \frac{du}{2\pi i} \frac{g(u)}{u-t}
%\end{gather}

All the kernels involved are of the form $K(z,w)$ with $z$ and $w$ on two disjoint curves, say $C_1$ and $C_2$. It is sufficient to check that $\iint_{C_1 \times C_2} |K(z,w)|^2 |dz| |dw| < \infty$ to ensure that the operator belongs to the class of Hilbert-Schmidt operators. This implies that $\left\{ \mathcal{A}_{j,1}\right\}_j$, $\mathcal{B}_1$ and $K^{GEN}_{a_j}$ are trace-class (for all $ j=1,\ldots, 2N$), since composition of two HS operators.
%Are these kernels HS?
%\begin{gather}
%\left\| \mathcal{O}_1 \right\|^2 = \int_{i\mathbb{R}+\delta} \frac{|dw|}{2\pi} \int_{\gamma} \frac{|d\xi|}{2\pi } \frac{e^{2\Re\left(\frac{\xi a}{2}\right) } }{|\xi -w|^2} < \infty \\
%\left\| \mathcal{O}_2 \right\|^2 = \int_{\hat \gamma} \frac{|ds|}{2\pi}  \int_{i\mathbb{R} + \delta} \frac{|dw|}{2\pi } \frac{ e^{2\Re\left(-as - \frac{\tau}{s} + \frac{1}{2s^2}\right) }}{|w-s|^2}|s|^{2\nu} < \infty
%\end{gather}
%and
%\begin{gather}
%\left\| \mathcal{P}_1\right\|^2 = \int_{i\mathbb{R}+\delta} \frac{|du|}{2\pi} \int_{\hat \gamma} \frac{|ds|}{2\pi } \frac{1}{|s-u|^2} < \infty \\
%\left\| \mathcal{P}_2 \right\|^2= \int_{\gamma} \frac{|dt|}{2\pi}   \int_{i\mathbb{R} + \delta} \frac{|du|}{2\pi } \frac{e^{2\Re\left(\frac{ta}{2} + \frac{\tau}{t} - \frac{1}{2t^2}\right)}}{|u-t|^2} |t|^{-2\nu} < \infty
%\end{gather}
%Now, we build up the integrable kernel, which will be trace-class, thanks to the fact that $\mathcal{A}$ and $\mathcal{B}$ are trace-class. 
\end{proof}
% := L^2(\gamma\cup \hat \gamma)\simeq L^2(\gamma) \oplus L^2(\hat \gamma) =

Now we recall that any operator acting on a Hilbert space of the type $H = H_1 \oplus H_2 $ can be decomposed as a $2\times 2$ matrix of operators with $(i,j)$-entry given by an operator $H_j \rightarrow H_i$. Thus, we can perform a chain of equalities
\begin{gather}
\det \left( \operatorname{Id}_{L^2(\gamma)} - \mathcal{K}^{GEN} \right) = \det \left( \operatorname{Id}_{L^2(\gamma)} -  \sum_{j=1}^{2N} (-1)^j \mathcal{B}_1 \circ \mathcal{A}_{j,1} \right) \nonumber \\
= \det \left( \operatorname{Id}_{L^2(\gamma)} \otimes \operatorname{Id}_{L^2(\hat \gamma)} -  \left[ \begin{array}{c|c}
0 & \mathcal{B}_1 \\ \hline
\sum_{j=1}^{2N} (-1)^j \mathcal{A}_{j,1}&0
\end{array} \right] \right) = \det \left( \operatorname{Id}_{L^2(\gamma \cup \hat \gamma)} - \mathbb{K}^{GEN}\right)
\end{gather}
the second equality follows from the multiplication on the left by the matrix (with determinant $=1$)
\begin{equation}
\text{Id}_{L^2(\gamma) \otimes L^2(\hat \gamma)} + \left[ \begin{array}{c|c} 
0 & -\mathcal{B}_1 \\ \hline 
0 & 0  
\end{array} \right]
\end{equation}
and the operator $\mathbb{K}^{GEN}$ is an integrable operator with kernel as in the statement.
\end{proof}

%\subsubsection{Fredholm determinant and IIKS integrable kernel}
%The Fredholm determinant is then

%And the integrable kernel $\mathbb{K}^{GEN}$ is the following: given the kernels
%\begin{gather}
%\mathcal{A}_{j,1}(t, s)=  \frac{e^{t \left( a_j- \frac{a_1}{2}\right) -a_js - \frac{\tau}{s} + \frac{1}{2s^2} }}{2\pi i(t -s)} s^\nu \cdot \chi_{\gamma}(t) \chi_{\hat \gamma}(s) \\
%\mathcal{B}_1(t,s) =  \frac{e^{ \frac{sa_1}{2} + \frac{\tau}{s} - \frac{1}{2s^2} }}{2\pi i (t-s)} s^{-\nu}\cdot \chi_{\gamma}(s) \chi_{\hat \gamma}(t)
%\end{gather}
%then
%\begin{gather}
%2\pi i (t-s)  \left[ \begin{array}{c|c}
%0 & \mathcal{B}_1 \\ \hline
%\sum_{j=1}^N (-1)^j \mathcal{A}_{j,1}&0
%\end{array} \right] = \nonumber \\
%\left[ \begin{array}{c|c}
%0&  e^{ \frac{sa_1}{2} + \frac{\tau}{s} - \frac{1}{2s^2} } s^{-\nu}\cdot \chi_{\gamma}(s) \chi_{\hat \gamma}(t) \\ \hline
%\sum_{j} (-1)^je^{t \left( a_j- \frac{a_1}{2}\right) -a_js - \frac{\tau}{s} + \frac{1}{2s^2} } s^\nu \cdot \chi_{\gamma}(t) \chi_{\hat \gamma}(s) &0
%\end{array}\right] 
%%\nonumber \\
%%\overset{*}{=} \left[\begin{array}{cc}\chi_{\Gamma}(t), & \chi_{i\mathbb{R}+\epsilon}(t) \end{array} \right] \cdot 
%%\left[ \begin{array}{c|c}
%%0&  e^{ \frac{sa_1}{2} + \frac{\tau}{s} - \frac{1}{2s^2} } s^{-\nu} \\ \hline
%%\sum_j (-1)^j e^{t \left( a_j- \frac{a_1}{2}\right) -a_js - \frac{\tau}{s} + \frac{1}{2s^2} } s^\nu  &0
%%\end{array}\right] 
%%\cdot \left[ \begin{array}{c} \chi_{\Gamma}(s) \\ \chi_{i\mathbb{R}+\epsilon}(t) \end{array}\right]
%\end{gather}

\subsection{Riemann-Hilbert problem and $\tau$-function}
We can proceed now with building a Riemann-Hilbert problem associated to the integrable kernel we just found in Theorem \ref{identityFD}. 
This will allow us to find some explicit identities for its Fredholm determinant.
 %with establishing a relation between he Fredholm determinant of the Generalized Bessel operator and a suitable Riemann-Hilbert problem. 

\begin{prop}
Given the integrable kernel (\ref{Kcritfg})-(\ref{Kcritg}), the correspondent RH-problem is the following: finding an $(2N+1)\times (2N+1)$ matrix $\Gamma$ such that it is analytic on $\mathbb{C} \backslash \Xi$ ($\Xi:= \gamma \cup \hat \gamma$) and
\begin{equation}
\left\{ \begin{array}{ll}
\Gamma_+(\lambda) =\Gamma_-(\lambda) M(\lambda) & \lambda \in \Xi \\
%:= \gamma \cup \hat \gamma \nonumber \\
\Gamma(\lambda) = I + \mathcal{O}(1/\lambda) & \lambda \rightarrow \infty
\end{array}
\right.
\end{equation}
with jump matrix $M(\lambda) := I - J(\lambda)$,
\begin{gather}
J (\lambda) := 2\pi i \vec{f}(\lambda) \cdot \vec{g}^T(\lambda)  \nonumber \\
= \left[ \begin{array}{ccccc}
0 &  - e^{ \theta_{a_1}}\chi_{\hat \gamma} & e^{\theta_{a_2}} \chi_{\hat \gamma}& \ldots & (-1)^{2N} e^{ \theta_{a_{2N}}}\chi_{\hat \gamma} \\
e^{ -\theta_{a_1}} \chi_{\gamma} & 0 & 0 & \ldots & 0 \\
e^{ -\theta_{a_2}}\chi_\gamma &0&\ldots&& \vdots \\
\vdots&&&& \\
e^{ -\theta_{a_{2N}}} \chi_\gamma & 0 & 0 & \ldots & 0 
\end{array}\right] \\
%+ \left[\begin{array}{ccccc} 
%0 & 0 & \ldots &  & 0 \\
%e^{ -\theta_{a_1}} &0 & \ldots && 0 \\
%e^{ -\theta_{a_2}} & 0 & \ldots && 0 \\
%\vdots &&&& \vdots \\
%e^{ -\theta_{a_N}} & 0 & \ldots && 0
%\end{array}\right] \chi_{\gamma} \\
\theta_{a_j} : = -a_j \lambda - \frac{\tau}{\lambda} + \frac{1}{2\lambda^2} + \nu \ln \lambda  \ \ \ \forall \, j=1, \ldots, 2N.
\end{gather}
\end{prop}

\begin{proof} We simply need to verify that $I-J(\lambda) = I - 2\pi i \vec f(\lambda) \cdot \vec{g} (\lambda)^T$.
\end{proof}

It is easy to see that the jump matrix is conjugate to a matrix with (piece-wise) constant entries
$M(\lambda) = e^{T(\lambda)} \cdot M_0 \cdot e^{-T(\lambda)} $ with
\begin{gather}
T(\lambda, \vec a) = \text{diag} \left(T_0, T_1, \ldots, T_N \right) \\
T_0 = \frac{1}{N+1} \sum_{j=1}^{2N} \theta_{a_j} \ \ \ \  T_j = T_0 - \theta_{a_j} 
%\theta_{a_j} : = -a_j \lambda - \frac{\tau}{\lambda} + \frac{1}{2\lambda^2} + \nu \ln \lambda  \ \ \ \forall j =1, \ldots, N
\end{gather}
with $\vec a$ the collection of all endpoints $\{a_j \}$.

Thus, considering the matrix $\Psi(\lambda, \vec{a}) := \Gamma (\lambda, \vec{a}) e^{T(\lambda, \vec{a})}$, $\Psi$ satisfies a RH-problem with constant jumps, thus it's (sectionally) a solution to a polynomial ODE.

Referring on the results stated in \cite{Misomonodromic} and \cite{MeM} and adapted to the case at hand, we can claim that
\begin{thm}\label{FDisomon}
For every parameter $\rho$, on which the Generalized Bessel operator may depend,
\begin{gather}
\partial_\rho \ln \det \left( \operatorname{Id} - \chi_I K^{GEN} \right) = \int_\Sigma \operatorname{Tr} \left( \Gamma^{-1}_{-}(\lambda) \Gamma'_{-}(\lambda) \Pi_{\partial_\rho}(\lambda) \right) \, \frac{d\lambda}{2\pi i}  \label{residue}
\end{gather}
where we recall that $I= \bigcup_k [a_{2k-1},a_{2k}]$ is the multi-interval, $\Xi = \gamma \cup \hat \gamma$ and $\Pi_{\partial_\rho}(\lambda) := \partial_\rho M (\lambda) M^{-1}(\lambda)$.
\end{thm}

Moreover, thanks to the Jimbo-Miwa-Ueno residue formula (see \cite{MeM}),
\begin{prop}
$\forall \, j=1,\ldots, 2N$ the Fredholm determinant satisfies
\begin{equation}
\partial_{a_j} \ln \det \left( \operatorname{Id} - \chi_I K^{GEN} \right)  = - \operatorname{res}_{\lambda=\infty} \operatorname{Tr}\left( \Gamma^{-1} \Gamma' \partial_{a_j}T    \right) =  \Gamma_{1; j+1, j+1} \label{residueJMU}
\end{equation}
i.e. the $(j+1, j+1)$ component of the residue matrix \,$\Gamma_1 = \lim_{\lambda \rightarrow \infty} \lambda \left(I - \Gamma(\lambda) \right) $.

As far as the parameter $\tau$ is concerned, the following result holds
\begin{equation}
\partial_{\tau} \ln \det \left( \operatorname{Id} - \chi_I K^{GEN} \right)  = \operatorname{res}_{\lambda=0} \operatorname{Tr}\left( \Gamma^{-1}\Gamma' \partial_{\tau}T    \right) = - \left( \tilde \Gamma_0^{-1} \tilde \Gamma_{1}\right)_{1,1}
\end{equation}
where $ \tilde\Gamma_0$ and $\tilde\Gamma_1$ are coefficients appearing in the asymptotic expansion of the matrix $\Gamma$ in a neighbourhood of zero.
\end{prop}
\begin{proof}
The phases $\theta_{a_j}$ are linear in $a_j$, exactly as in the Bessel kernel case (see \cite{Me}). 
\begin{equation}
 \partial_{a_j}T(\lambda, \vec a) =  \lambda \left( \frac{1}{2N+1}I - E_{j+1, j+1} \right)
\end{equation}
Then, we plug this expression into (\ref{residueJMU})
\begin{equation}
\operatorname{res}_{\lambda =\infty} \operatorname{Tr} \left(\Gamma^{-1}\Gamma' \partial_{a_j}T \right) = \frac{\operatorname{Tr} \Gamma_1}{2N+1}  - \Gamma_{1; j+1,j+1}
\end{equation}

Regarding the residue at zero, we recall the asymptotic expansion of $\Gamma \sim  \tilde \Gamma_0 + \lambda \tilde \Gamma_1 + \ldots$ near zero (see \cite{Wasow}) and we calculate 
\begin{gather}
\partial_\tau T  = - \frac{1}{\lambda} \left[ E_{1,1} -\frac{1}{2N+1}I  \right]
\end{gather}
thus
\begin{equation}
\operatorname{res}_{\lambda =0} \operatorname{Tr} \left(\Gamma^{-1}\Gamma' \partial_{\tau}T \right) = \frac{\operatorname{Tr} \left(\tilde \Gamma_0^{-1} \tilde \Gamma_1\right)}{2N+1}  - \left( \tilde \Gamma_0^{-1} \tilde \Gamma_1\right)_{1,1}
\end{equation}
The result follows from $\text{Tr} \Gamma_1 = \text{Tr}   \left(\tilde \Gamma_0^{-1} \tilde \Gamma_{1} \right)= 0$, since $\det \Gamma(\lambda) \equiv 1$.
\end{proof}

\subsection{The single-interval case}\label{1timeGENB0a}
In case we consider a single interval $I=[0,a]$, we are able to perform a deeper analysis on the gap probability of the Generalized Bessel operator and link it to an explicit Lax pair.

We will see that the Lax pair $A$ and $U$ will recall the Bessel Lax pair very closely (see \cite{Me}), except for the presence of an extra term for the spectral matrix $A$. Such term will introduce a higher order Poincar\'e rank at $\lambda=0$ as it will be clear in the following calculations. Moreover, thanks to the presence of the parameter $\tau$ other than the endpoint $a$, we can actually calculate a Lax ``triplet".

First of all, we reformulate Theorems \ref{identityFD} and \ref{FDisomon}, focusing on our present case.
\begin{thm}
Given $I=[0,a]$, the following equality between Fredholm determinants holds
\begin{gather}
\det \left( I_{L^2(\gamma) }- \chi_{[0,a]}K^{GEN}\right) 
%= \det \left( I_{L^2(\gamma)} - \mathcal{B}\circ \mathcal{A}\right) \nonumber \\
%= \det \left( I_{L^2(\gamma)}\otimes I_{L^2(\hat \gamma)} - \left[ \begin{array}{c|c} 0& \mathcal{B}\\\hline \mathcal{A} & 0\end{array}\right]\right) 
= \det \left( I_{L^2(\gamma \cup\hat \gamma)} - \mathbb{K}^{GEN} \right)
\end{gather}
with $\mathbb{K}^{GEN}$ an IIKS integrable operator with kernel
\begin{align}
&\mathbb{K}^{GEN}_{\nu,\tau}(t,s) 
%= \frac{ e^{\frac{t a}{2} -as - \frac{\tau}{s} + \frac{1}{2s^2}}s^\nu \chi_{\gamma}(t) \chi_{\hat \gamma}(s) + e^{\frac{s a}{2} + \frac{\tau}{s} - \frac{1}{2s^2}} s^{- \nu} \chi_{\hat \gamma}(t) \chi_{\gamma} (s)}{2\pi i( t-s)}  \nonumber \\
= \frac{\vec{f}^T(t) \cdot \vec{g}(s)}{t-s} \label{Kcritint} \\
&\vec{f}(t) = \frac{1}{2\pi i} \left[ \begin{array}{c} e^{\frac{ta}{2}}  \\ 0 \end{array} \right]\cdot \chi_{\gamma}(t) +  \frac{1}{2\pi i} \left[ \begin{array}{c} 0 \\ 1 \end{array} \right]\cdot \chi_{\hat \gamma}(t)  \\
&\vec{g}(s) = \left[ \begin{array}{c}  0 \\  s^{-\nu} e^{\frac{sa}{2} + \frac{\tau}{s} - \frac{1}{2s^2}} \end{array} \right]\cdot \chi_{\gamma}(s) +  \left[ \begin{array}{c} s^\nu e^{-sa - \frac{\tau}{s} + \frac{1}{2s^2}} \\ 0 \end{array} \right]\cdot \chi_{\hat \gamma}(s)
\end{align}
\end{thm}

%\begin{note}
%It's a regular kernel, since for $t=s$ the numerator vanishes.
%\end{note}

The associated RH-problem reads as follows:
\begin{equation}
\left\{ \begin{array}{ll}
\Gamma_+(\lambda) = \Gamma_-(\lambda) M(\lambda) & \lambda \in \Xi:= \hat \gamma \cup \gamma \nonumber \\
\Gamma(\lambda) = I + \mathcal{O}(1/\lambda) & \lambda \rightarrow \infty
\end{array}
\right.
\end{equation}
with $\Gamma$ a $2\times 2$ matrix, analytic on analytic on the complex plane except on the collection of curves $\Xi$, along which the above jump condition is satisfied with  jump matrix $M(\lambda) := I - J(\lambda)$
\begin{gather}
M(\lambda) = \left[ \begin{array}{cc} 1 &  -e^{\lambda a + \frac{\tau}{\lambda}  - \frac{1}{2\lambda^2} -\nu \ln \lambda}   \chi_{\gamma}(\lambda)\\
-e^{-\lambda a - \frac{\tau}{\lambda} + \frac{1}{2\lambda^2} + \nu \ln \lambda}\chi_{\hat \gamma} (\lambda)& 1 \end{array} \right]  \nonumber \\
= e^{T_a(\lambda)} \cdot M_0 \cdot e^{-T_a(\lambda)}
\end{gather}
Thus the jump matrix $M$ is equivalent to a matrix with constant entries, via the conjugation $e^{T_a(\lambda)}$, $T_a(\lambda) = \frac{1}{2}\theta_a  \sigma_3$, where $\theta_a := - \lambda a - \frac{\tau}{\lambda} +  \frac{1}{2\lambda^2} + \nu\ln \lambda$ and $\sigma_3$ is the third Pauli matrix. This allows us to define the matrix $\Psi (\lambda):= \Gamma(\lambda) e^{T_a(\lambda)}$ which solves a RHP with constant jumps and is (sectionally) a solution to a polynomial ODE:
%\begin{gather}
%\begin{array}{ll}
%\Psi_+(\lambda) = \Psi_-(\lambda) M_0 & \text{on} \ \Xi \\
%\Psi(\lambda) = \left(I - \mathcal{O}\left(\frac{1}{\lambda}\right) \right) e^{\frac{1}{2}\theta_a \sigma_3} &\text{as} \ \lambda \rightarrow \infty
%\end{array}
%\end{gather}

%\subsection{The JMU $\tau$-function and our Fredholm determinant}

Applying Theorem \ref{FDisomon}, we get
\begin{thm}
\begin{gather}
\partial_\rho \ln \det (I - \chi_{[0,a]}K^{GEN}) =  \int_{\Xi} \operatorname{Tr} \left( \Gamma_-^{-1}(\lambda) \Gamma'_-(\lambda) \Pi_{\partial_\rho} (\lambda) \right) \frac{d\lambda}{2\pi i} \\
\Pi_\partial (\lambda) := \partial M(\lambda) M^{-1}(\lambda),  \ \ \ \Xi:= \gamma \cup \hat \gamma
\end{gather}
for every parameter $\rho$ on which the operator $K^{GEN}$ depends.
%, with $\omega$ the logarithmic derivative of the JMU isomonodromic $\tau$-function (remember that $\Xi:= \gamma \cup \hat \gamma$).
\end{thm}

In particular, thanks to the Jimbo-Miwa-Ueno residue formula, we have
\begin{align}
\partial_a \ln \det (I - \chi_{[0,a]}K^{GEN})  &=- \operatorname{res}_{\lambda =\infty} \ \operatorname{Tr} \left( \Gamma^{-1} \Gamma' \partial_aT_a \right) \\
\partial_\tau \ln \det (I - \chi_{[0,a]}K^{GEN})  &= \operatorname{res}_{ \lambda = 0} \ \operatorname{Tr} \left( \Gamma^{-1} \Gamma' \partial_\tau T_a \right)
\end{align}
%(the only pole is at infinity, since the phase is $\theta_a = \lambda a + \text{constant terms in $a$}$).

\begin{prop}The Fredholm determinant of the Generalized Bessel operator satisfies the following relations
\begin{align}
\partial_a \ln \det (\operatorname{Id} - \chi_{[0,a]}K^{GEN}) &=  \Gamma_{1;2,2} \label{O122} \\
\partial_\tau \ln \det (\operatorname{Id} - \chi_{[0,a]}K^{GEN}) &= \left( \tilde \Gamma_0^{-1} \tilde \Gamma_1 \right)_{2,2}
\end{align}
with $\Gamma_{1;2,2}$ the $(2,2)$-entry of the residue matrix $\Gamma_1$ at $\infty$, while the $ \tilde \Gamma_j$'s appear in the asymptotic expansion of $\Gamma$ near zero.
\end{prop}

We can now calculate the Lax ``triplet" associated to the RH-problem above:
\begin{align}
A&:= \partial_\lambda \Psi \cdot \Psi^{-1} =A_0 + \frac{A_{-1}}{\lambda}+ \frac{A_{-2}}{\lambda^2}+ \frac{A_{-3}}{\lambda^3}  \\
U&:= \partial_a \Psi \cdot \Psi^{-1} = U_0 + \lambda U_1 \\
V&:= \partial_\tau \Psi \cdot \Psi^{-1}  = V = V_0 + \frac{V_{-1}}{\lambda}
\end{align}
%\begin{gather}
%\partial_a \Psi = U \Psi, \ \ \ 
%\partial_\tau \Psi = V \Psi, \ \ \ 
%\partial_\lambda \Psi = A \Psi,
%\end{gather}
with coefficients
\begin{align}
A_{0} &= \frac{a}{2}\sigma_3,  \ \ \ A_{-1}  = - \frac{\nu}{2}\sigma_3 + \frac{a}{2}[\Gamma_1, \sigma_3] \\
A_{-2} &=  - \frac{a}{2} [\Gamma_1,\sigma_3\Gamma_1] + \frac{a}{2} [\Gamma_2,\sigma_3]  - \frac{\nu}{2}[\Gamma_1,\sigma_3] - \frac{\tau}{2}\sigma_3  - \Gamma_1 \\
A_{-3}&= \Gamma_1^2 - 2\Gamma_2 +\frac{a}{2}[\sigma_3\Gamma_2,\Gamma_1] + \frac{a}{2}[\Gamma_1,\sigma_{3}\Gamma_1^2] + \frac{a}{2}[\sigma_3\Gamma_1,\Gamma_2] + \frac{a}{2}[\Gamma_3,\sigma_3]   \nonumber \\ 
&+\frac{\nu}{2}\sigma_3\Gamma_2   + \frac{\nu}{2} [\Gamma_1,\sigma_3\Gamma_1] + \frac{\tau}{2} \sigma_3\Gamma_1 + \frac{1}{2}\sigma_3  \\
U_0&=  \frac{1}{2} [\Gamma_1,\sigma_3], \ \ \ U_1  = \frac{1}{2}\sigma_3 \\
V_0 &= 0 , \ \ \ V_{-1} = \frac{1}{2}\sigma_3
\end{align}

We point out that $\lambda =0$ is an irregular point of Poincar\'e rank $2$. The behaviour at zero shows a higher order rank with respect to the Lax pair for the Bessel operator \cite{Me}, where the point $\lambda=0$ was of rank $1$. 

Moreover, the matrix $U$ is the same as the one appearing in the Bessel Lax pair (in the non-rescaled case, see \cite{Me}). 

The expression of the Lax pair $A$ and $U$ suggests that their compatibility equation will eventually lead to a higher order ODE belonging to the Painlev\'e III hierarchy.

%The compatibility equations are
%\begin{gather}
%\partial_\lambda U  - \partial_a A + [U,A] = 0 \\
%\partial_\lambda V  - \partial_\tau A + [V,A] = 0 \\
%\partial_\tau U  - \partial_a V + [U,V] = 0 
%\end{gather}

%{\color{red} More work to do this...?}

\section{Multi-time Generalized Bessel}\label{ntimeGENB}

%\subsection{Conjugation and Fourier Transform}
%Balint's kernel, although I'm pretty sure it's equivalent to mine, looks more accommodating because it recalls even more the multi-time Bessel. Morevoer, it looks like that when we need to explode the kernel this way is easier (we can recover the diagonal sub-kernel $\mathcal{N}$ as in the multi Bessel).

%The kernel I'm starting with is (check on "Matching.tex")

The multi-time Generalized Bessel operator on $L^2(\mathbb{R}_+)$ with times $\tau_1<\ldots< \tau_n$ is defined through a matrix kernel with entries $K^{GEN}:= H_{ij} + \chi_{i<j} P_{\Delta_{ij}}$
\begin{align}
H_{ij}(x,y)& = - 4 \left(  \frac{y}{x}\right)^\nu \int_{\gamma \times \hat \gamma} \frac{dt \, ds}{(2\pi i)^2} \frac{e^{ - \frac{1}{2}\left(\tau - \frac{1}{t} \right)^2 + xt + \frac{1}{2}\left(\tau - \frac{1}{s}+ \Delta_{ji}\right)^2 - ys}}{\left(s-t+ \Delta_{ji} ts\right)} \left(\frac{s}{t}\right)^\nu \\
P_{\Delta_{ij}} (x,y) &= \left( \frac{y}{x} \right)^{\frac{\nu}{2}} \frac{1}{\Delta_{ji}} e^{-\frac{x+y}{4\Delta_{ji}}} I_\nu\left( \frac{\sqrt{xy}}{2\Delta_{ji}} \right) \nonumber \\
&= - \left( \frac{y}{x}\right)^{\nu} \frac{1}{\Delta_{ji}}   \int_\gamma e^{\frac{x}{4\Delta_{ji}}(t-1) + \frac{y}{4\Delta_{ji} }\left( \frac{1}{t} -1 \right)} t^{-\nu -1} \frac{dt}{2\pi i}
\end{align}
the curve $\gamma$ is the same one as in the single-time Extended Bessel kernel (a contour that winds around zero counterclockwise an extends to $-\infty$) and $\hat\gamma := \frac{1}{\gamma}$; $\Delta_{ji}:= \tau_j-\tau_i>0$.
%$\hat \gamma_j := \frac{1}{\gamma + 4\tau_j} $, $\forall \, j=1,\ldots, n$; $\Delta_{ji}:= \tau_j-\tau_i$
\begin{oss}
The matrix $P_{\Delta_{ij}}$ is strictly upper triangular. 
\end{oss}

\begin{oss}
The above definition of the multi-time kernel is the one given by Delvaux and Veto (\cite{Balint}). We preferred to use this one because the study of the gap probability with this expression involves less complicated calculations than with our equivalent version (see Appendix). 
\end{oss}

As in the single-time case, we are again interested in the gap probability of the operator restricted to a collection intervals $I_j$ at each time $\tau_j$ ($\forall \, j$), i.e. 
\begin{equation}
\det\left( \text{Id}_{L^2(\mathbb{R}_+)} - K^{GEN} \chi_{\mathcal{I}} \right)
\end{equation}
where $\chi_{\mathcal{I}} = \text{diag} \left( \chi_{I_1}, \ldots, \chi_{I_n} \right)$ is a diagonal matrix of characteristic functions  and
%$\mathcal{I} = \{ I_1, \ldots, I_n \}$, and 
\begin{equation*}
I_j:= [a_1^{(j)}, a_2^{(j)}] \cup \ldots \cup [a_{2k_j-1}^{(j)}, a_{2k_j}^{(j)}] \ \ \ \forall \, j=1,\ldots,n.
\end{equation*}
%$k_j$ even for all $j$'s.

\begin{oss}
The multi-time Bessel operator fails to be trace-class on infinite intervals.
\end{oss}

For the sake of clarity, we will focus on the simple case $I_j= [0, a^{(j)}]$, $\forall \, j$.  The general case follows the same guidelines described below; the only difficulties are mostly technical, due to heavy notation, and not theoretical.
% unnecessary notational complications. {\color{red} say that the general case is just mere calculations}

As in the single-time case, we start by establishing a link between the multi-time Generalized Bessel operator and a suitable IIKS operator, which we will examine deeper in the next subsection.
\begin{thm}
The following identity between Fredholm determinants holds
\begin{equation}
\det \left( \operatorname{Id} -K^{GEN}\chi_{\mathcal{I}} \right) = \det \left( \operatorname{Id} - \mathbb{K}^{GEN}  \right)
\end{equation}
with $\chi_{\mathcal{I}} = \operatorname{diag} \, \left(\chi_{I_1}, \ldots, \chi_{I_n} \right)$
the characteristic matrix of the collection of intervals. The operator $\mathbb{K}_B$ is an integrable operator acting on the Hilbert space
\begin{equation}
H: =L^2\left(\gamma \cup \bigcup_{k=1}^n \gamma_{-k} ,\mathbb{C}^n\right) \sim  L^2\left(\bigcup_{k=1}^n\gamma_{-k}, \mathbb{C}^n\right) \oplus  L^2(\gamma, \mathbb{C}^n),
\end{equation}
with $\gamma_{-k} := \frac{1}{\gamma}-4\tau_k$.

Its kernel is a $2n \times 2n$ matrix  of the form
\begin{align}
&\mathbb{K}^{GEN} (v,\xi) = \frac{\textbf{f}(v)^T \cdot \textbf{g}(\xi)}{v-\xi}  \label{IIKSGENBesselmulti1} \\
&\textbf{f}(v)^T = \frac{1}{2\pi i} \left[ \begin{array}{c|c|c}
\operatorname{diag}\, \mathcal{N}(v) & 0 & 0 \\ \hline
0 & \operatorname{diag} \, \mathcal{M}(v) & \mathcal{A}(v)
\end{array} \right] \\
&\textbf{g}(\xi) = \left[ \begin{array}{c|c}
0 & \operatorname{diag}\, \mathcal{N}(\xi) \\ \hline
 \mathcal{M}(\xi) & 0 \\ \hline
0 & \mathcal{B}(\xi)
\end{array}\right]
\end{align}
where $\textbf{f}, \textbf{g}$ are $N \times 2n$ matrices, with $N = 2n + (n-1) = 3n-1$.
\begin{gather}
\operatorname{diag} \, \mathcal{N}(v) := \operatorname{diag}  \left[ -4 e^{-\frac{a^{(1)}}{v_1}}\chi_{\gamma}, \ldots, -4 e^{-\frac{a^{(n)}}{v_n}}\chi_{\gamma} \right] \\
\operatorname{diag} \, \mathcal{N}(\xi) := \operatorname{diag} \left[ e^{\frac{a^{(1)}}{\xi_1}}\chi_{\gamma_{-1}}, \ldots,  e^{-\frac{a^{(n)}}{\xi_n}}\chi_{\gamma_{-n}} \right] \\
 \operatorname{diag} \, \mathcal{M}(v) := \operatorname{diag} \left[ e^{- \frac{(v_{1,\tau})^2}{2}}v_1^\nu \chi_{\gamma_{-1}}, \ldots, e^{- \frac{(v_{n,\tau})^2}{2}}v_n^\nu \chi_{\gamma_{-n}}\right] \\
 \mathcal{M}(\xi) := \left[ \begin{array}{ccc}
 e^{  \frac{(\xi_{1,\tau})^2}{2}} \xi_1^{-\nu} \chi_{\gamma}  & \ldots & e^{  \frac{(\xi_{1,\tau})^2}{2}} \xi_n^{-\nu} \chi_{\gamma}  \\
 \vdots & & \vdots \\
 e^{  \frac{(\xi_{n,\tau})^2}{2}} \xi_1^{-\nu} \chi_{\gamma}  & \ldots & e^{  \frac{(\xi_{n,\tau})^2}{2}} \xi_n^{-\nu} \chi_{\gamma} 
 \end{array} \right]
\end{gather}
\begin{gather}
\mathcal{A}(v) =\nonumber \\
\left[ \begin{array}{ccccc}
 -4 e^{-\frac{a^{(2)}}{v_2}} \frac{v_1^\nu}{v_2^\nu}\chi_{\gamma_{-1}} & -4 e^{-\frac{a^{(3)}}{v_3}} \frac{v_1^\nu}{v_3^\nu}\chi_{\gamma_{-1}} & -4 e^{-\frac{a^{(4)}}{v_4}} \frac{v_1^\nu}{v_4^\nu}\chi_{\gamma_{-1}} & \ldots & -4 e^{-\frac{a^{(n)}}{v_n}} \frac{v_1^\nu}{v_n^\nu}\chi_{\gamma_{-1}} \\
 0 & -4 e^{-\frac{a^{(3)}}{v_3}} \frac{v_2^\nu}{v_3^\nu}\chi_{\gamma_{-2}} & -4 e^{-\frac{a^{(4)}}{v_4}} \frac{v_2^\nu}{v_4^\nu}\chi_{\gamma_{-2}} & \ldots & -4 e^{-\frac{a^{(n)}}{v_n}} \frac{v_2^\nu}{v_n^\nu}\chi_{\gamma_{-2}} \\
  & 0 & -4 e^{-\frac{a^{(4)}}{v_4}} \frac{v_3^\nu}{v_4^\nu}\chi_{\gamma_{-3}} & \ldots & \\
&& \vdots && \\ 
&&& 0 & -4 e^{-\frac{a^{(n)}}{v_n}} \frac{v_{n-1}^\nu}{v_n^\nu}\chi_{\gamma_{-{(n-1)}}} \\
&&&& 0
\end{array} \right] \\
\mathcal{B}(\xi) =
\left[ \begin{array}{cccccc}
0 & e^{\frac{a^{(2)}}{\xi_2}}\chi_{\gamma_{-2}} &&&& \\
& 0 & e^{\frac{a^{(3)}}{\xi_3}}\chi_{\gamma_{-3}} &&& \\
&& 0 & e^{\frac{a^{(4)}}{\xi_4}}\chi_{\gamma_{-4}} && \\
&&&& \ddots & \\
&&&& 0 & e^{\frac{a^{(n)}}{\xi_n}}\chi_{\gamma_{-n}} 
\end{array} \right] \label{IIKSGENBesselmulti2}
\end{gather}
$\zeta_k := \zeta + 4\tau_k$, $\zeta_{k,\tau}:= \zeta+4\tau_k-\tau$ ($\zeta = v, \xi$, $k=1,\ldots,n$).
\end{thm}

\begin{oss}
By Fredholm determinant $``\det"$ we denote the determinant defined through the usual series expansion
\begin{equation}
\det(\operatorname{Id}-K) := 1+ \sum_{k=1}^\infty \frac{1}{k!}\int_{X^k} \det [K(x_i,x_j)]_{i,j=1}^k d\mu(x_1)\ldots d\mu(x_k) \label{Freddet}
\end{equation} 
with $K$ an integral operator acting on the Hilbert space $L^2(X, d\mu(x))$ and kernel $K(x,y)$.

In the case at hand, we will see that the operator $K^{GEN}\chi_{\mathcal{I}}= (H + P_{\Delta})\chi_{\mathcal{I}}$ is the sum of a trace-class operator ($H \chi_{\mathcal{I}}$) plus a Hilbert-Schmidt operator ($P_\Delta \chi_{\mathcal{I}}$) with diagonal-free kernel.

Therefore the naming of Fredholm determinant refers to the following expression:
\begin{gather}
``\det" (\operatorname{Id} - K^{GEN}\chi_{\mathcal{I}}) 
%= ``\det" (\operatorname{Id} - H\chi_{\mathcal{I}} - P_\Delta\chi_{\mathcal{I}}) \nonumber \\
= e^{\operatorname{Tr} H} \left.\det\right._2(\operatorname{Id} - K^{GEN}\chi_{\mathcal{I}})
\end{gather}
where $\left. \det \right._2$ denotes the regularized Carleman determinant (see \cite{traceideals}).
\end{oss}

\begin{proof}
Thanks to the invariance of the Fredholm determinant under kernel conjugation, we can discard the term $\left( \frac{y}{x} \right)^{\nu}$ in our further calculations.

We will work on the entry $(i,j)$ of the kernel. We can notice that for $x<0$ or $y<0$ the kernel is identically zero, $K^{GEN}(x,y) \equiv 0$. Then, applying Cauchy's theorem, we have
\begin{gather}
\chi_{[0,a^{(j)}]}(y) H_{ij}(x,y)  \nonumber \\
%=  H_{ij}(x,y) \chi_{(-\infty,a^{(j)}]}(y) \nonumber \\
= 4  \int_{i\mathbb{R} + \epsilon} \frac{d\xi}{2\pi i} \frac{e^{\xi(a^{(j)}-y)}}{\xi-s}  \iint_{\hat \gamma \times \gamma} \frac{ds \, dt}{(2\pi i)^2}  \frac{e^{  - a^{(j)}s + xt + \frac{1}{2}\left( \tau-\frac{1}{s} +4\Delta_{ji}  \right)^2 - \frac{1}{2} \left( \tau-\frac{1}{t} \right)^2 }}{ \left(\frac{1}{s} - \frac{1}{t} - 4\Delta_{ji}  \right)} \left( \frac{s}{t} \right)^\nu \frac{1}{st} \nonumber \\
=  -4  \int_{i\mathbb{R} + \epsilon} \frac{d\xi}{2\pi i} e^{-y \xi}  \int_{i\mathbb{R} + \epsilon} \frac{ dt}{2\pi i}  e^{xt} \int_\gamma \frac{dv}{2\pi i} \frac{e^{ a^{(j)}\xi - \frac{1}{2} \left( \tau-\frac{1}{t} \right)^2 - \frac{a^{(j)}}{v+4\tau_j} + \frac{1}{2}\left( \tau - 4\tau_{i} -v  \right)^2 }}{ \left(  \frac{1}{\xi} -4\tau_j-v \right)\left( \frac{1}{t} - 4\tau_{i}-v  \right)} \left( \frac{1}{(v+4\tau_j)t} \right)^{\nu} \frac{1}{\xi t}
%\int_{\hat \gamma} \frac{ds}{2\pi i} \frac{e^{ a^{(j)}\xi - \frac{1}{2} \left( \tau-\frac{1}{t} \right)^2 - a^{(j)}s + \frac{1}{2}\left( \tau-\frac{1}{s} +4\Delta_{ji}  \right)^2 }}{ (\xi - s)\left(\frac{1}{s} - \frac{1}{t} - 4\Delta_{ji}  \right)} \left( \frac{s}{t} \right)^\nu \frac{1}{st} 
\end{gather}
where we deformed $\gamma$ into a translated imaginary axis $i\mathbb{R}+\epsilon$ ($\epsilon>0$) in order to make Fourier operator defined below more explicit; the last equality follows from the change of variable on $s = 1/(v+4\tau_j)$, thus the contour $\hat \gamma$ becomes similar to $\gamma$ and can be continuously deformed into it. 
%{\color{red} Regarding the resulting curve, check on the multi-time Bessel. Starting from the curve $\hat \gamma_j := \frac{1}{\gamma + 4\tau_j} $, instead of just $\hat \gamma$ (everything should still work), the resulting curve will be $\gamma$. CHECK!!! Especially the cuts.}

On the other hand,
\begin{gather}
 \chi_{[0, a^{(j)}]}(y)  P_{\Delta_{ji}}(x,y) \nonumber \\
= \frac{-1}{\Delta_{ji}}  \int_{i\mathbb{R}+\epsilon} \frac{d\xi}{2\pi i} e^{-\xi y}  \int_{ \gamma} \frac{e^{\xi a^{(j)} + \frac{x}{4\Delta_{ji}}(t-1) - \frac{a^{(j)}}{4\Delta_{ji} }\left(1- \frac{1}{t} \right)}}{\xi - \frac{1}{4\Delta_{ji}}\left(1-\frac{1}{t}\right)} t^{-\nu -1} \frac{dt}{2\pi i}  \nonumber \\
=  -4 \int_{i\mathbb{R}+\epsilon} \frac{d\xi}{2\pi i} e^{-\xi y}   \int_{i\mathbb{R}+ \epsilon} \frac{dt}{2\pi i}   e^{xt}  \frac{e^{ a^{(j)} \left(\xi - \frac{t}{4\Delta_{ji}t+1 }\right)}}{t \xi \left( 4\Delta_{ji} + \frac{1}{t} -\frac{1}{\xi} \right)} (4\Delta_{ji}t+1)^{-\nu} 
\end{gather}
%change of variable: $\tilde t := \frac{t-1}{4\Delta_{ji}}$
%\begin{gather}
%= \frac{-1}{\Delta_{ji}}  \int_{i\mathbb{R}+\epsilon_1} \frac{d\xi}{2\pi i} e^{-\xi y}  \int_{i\mathbb{R}+\tilde \epsilon_2} \frac{e^{\xi a^{(j)} +xt - \frac{a^{(j)}t}{4\Delta_{ji}t+1 }}}{\xi - \frac{t}{4\Delta_{ji}t+1}} (4\Delta_{ji}t+1)^{-\nu -1} 4\Delta_{ji} \frac{dt}{2\pi i} \nonumber \\
%= -4 \int_{i\mathbb{R}+\epsilon_1} \frac{d\xi}{2\pi i} e^{-\xi y} \cdot   \frac{e^{ a^{(j)} \left(\xi - \frac{t}{4\Delta_{ji}t+1 }\right)}}{t \xi \left( 4\Delta_{ji} + \frac{1}{t} -\frac{1}{\xi} \right)} (4\Delta_{ji}t+1)^{-\nu} \cdot \int_{i\mathbb{R}+\tilde \epsilon_2}  e^{xt} \frac{dt}{2\pi i} 
%\end{gather}

It is easily recognizable the conjugation with a Fourier-like operator as in (\ref{Fourier}), so that
\begin{equation}
\left(K^{GEN}\chi_{\mathcal{I}}\right)_{ij} = \mathcal{F}^{-1} \circ \left( \mathcal{H}_{ij} + \chi_{i<j}\mathcal{P}_{ij} \right) \circ \mathcal{F}
\end{equation}
with
\begin{gather}
\mathcal{H}_{ij}(\xi,t):=  -4\int_\gamma \frac{dv}{2\pi i} \frac{e^{ a^{(j)}\xi - \frac{1}{2} \left( \tau-\frac{1}{t} \right)^2 - \frac{a^{(j)}}{v+4\tau_j} + \frac{1}{2}\left( \tau - 4\tau_{i} -v  \right)^2 }}{ \left(  \frac{1}{\xi} -4\tau_j-v \right)\left( \frac{1}{t} - 4\tau_{i}-v  \right)} \left( \frac{1}{(v+4\tau_j)t} \right)^{\nu} \frac{1}{\xi t} \\
\mathcal{P}_{ij}(\xi, t):= -4  \frac{e^{ a^{(j)} \left(\xi - \frac{t}{4\Delta_{ji}t+1 }\right)}}{  4\tau_{j} - 4\tau_i + \frac{1}{t} -\frac{1}{\xi}} (4\Delta_{ji}t+1)^{-\nu} \frac{1}{\xi t}
\end{gather}

Now we can perform the following change of variables on the Fourier-transformed kernel
\begin{gather}
\xi_j:= \frac{1}{\xi} - 4\tau_j, \ \ \  \eta_i:= \frac{1}{t} -4 \tau_i \label{changevar}
\end{gather}
so that the kernel will have the final expression
\begin{gather}
 \mathcal{K}^{GEN}_{ij}(\xi,\eta) =  \mathcal{H}_{ij}+ \chi_{\tau_i<\tau_j}\mathcal{P}_{ij} = \nonumber \\
-4\int_{\gamma} \frac{dv}{2\pi i} \frac{e^{ \frac{a^{(j)}}{\xi+4\tau_j} - \frac{1}{2} \left( \tau -4\tau_i -\eta \right)^2 - \frac{a^{(j)}}{v+4\tau_j} + \frac{1}{2}\left( \tau - 4\tau_{i} -v  \right)^2 }}{ \left(  \xi-v \right)\left( \eta-v  \right)} \left( \frac{\eta+4\tau_i}{v+4\tau_j} \right)^{\nu}  \nonumber \\
+ 4\chi_{\tau_i<\tau_j}   \frac{e^{ \frac{a^{(j)}}{\xi + 4\tau_j} - \frac{a^{(j)}}{\eta + 4\tau_j }}}{\xi-\eta} \left(\frac{4\Delta_{ji}}{\eta + 4\tau_i}+1\right)^{-\nu} 
\end{gather}
with $\xi \in \frac{1}{\gamma} - 4\tau_j =: \gamma_{-j} $ and $\eta \in \frac{1}{\gamma} - 4\tau_i=: \gamma_{-i}$.
The obtained (Fourier-transformed) Generalized Bessel operator is an operator acting on $L^2\left(\bigcup_{k=1}^n \gamma_{-k}, \mathbb{C}^n \right) \sim \bigoplus_{k=1}^n L^2\left( \gamma_{-k}, \mathbb{C}^n \right)$.

\begin{lemma}
The following decomposition holds $\mathcal{K}^{GEN} = \mathcal{M}\circ\mathcal{N} + \mathcal{P}$, with $\mathcal{M}$, $\mathcal{N}$, $\mathcal{P}$ Hilbert-Schmidt operators
\begin{gather}
\mathcal{M} : L^2\left(\gamma, \mathbb{C}^n \right) \rightarrow L^2\left(\bigcup_{k=1}^n \gamma_{-k}, \mathbb{C}^n  \right) \\
\mathcal{N} : L^2\left(\bigcup_{k=1}^n \gamma_{-k}, \mathbb{C}^n  \right) \rightarrow L^2\left(\gamma, \mathbb{C}^n \right) \\
\mathcal{P} : L^2\left(\bigcup_{k=1}^n \gamma_{-k}, \mathbb{C}^n  \right) \rightarrow L^2\left(\bigcup_{k=1}^n \gamma_{-k}, \mathbb{C}^n  \right)
\end{gather}
with kernel entries
\begin{gather}
\mathcal{M}_{ij}(v, \eta) = \frac{e^{ - \frac{1}{2} \left( \tau -4\tau_i -\eta \right)^2 + \frac{1}{2}\left( \tau - 4\tau_{i} -v  \right)^2 }}{\left( \eta-v  \right)} \left( \frac{\eta+4\tau_i}{v+4\tau_j} \right)^{\nu}    \chi_{\gamma}(v) \chi_{\gamma_{-i}}(\eta)\\
\mathcal{N}_{ij}(\xi,v; a^{(j)}) = 4\delta_{ij} \frac{ e^{  a^{(j)}\left(  \frac{1}{\xi+4\tau_j} - \frac{1}{v+4\tau_j}\right) }}{ \xi - v } \chi_{\gamma_{-j}}(\xi) \chi_{\gamma}(v) \\
\mathcal{P}_{ij} (\xi,\eta; a^{(j)})= 4\chi_{\tau_i<\tau_j}   \frac{e^{ \frac{a^{(j)}}{\xi + 4\tau_j} - \frac{a^{(j)}}{\eta + 4\tau_j }}}{\xi-\eta} \left(\frac{\eta + 4\tau_j}{\eta + 4\tau_i}\right)^{-\nu} \chi_{\gamma_{-i}}(\eta) \chi_{\gamma_{-j}}(\xi)
\end{gather}
\end{lemma}

\begin{proof} As in Lemma \ref{lemma5}, all the kernels involved are of the form $K(z,w)$ with $z$ and $w$ on two disjoint curves, say $C_1$ and $C_2$. The Hilbert-Schmidt property it thus ensured by simply checking that $\iint_{C_1 \times C_2} |K(z,w)|^2 |dz| |dw| < \infty$.
%Regarding the kernels $\mathcal{P}$ and $\mathcal{N}$, we already know that thy are HS (check on "MultiTBessel.tex"). As for the kernel $\mathcal{M}$, we check the $L^2$-norm:
%\begin{gather}
%\left\| \mathcal{M}_{ij} \right\| = \int_{\gamma_{-i}} \frac{|d\eta|}{2\pi} \int_{\gamma} \frac{|dt|}{2\pi} \,  \frac{ e^{2 \Re{ \left[ \frac{1}{2}\left( \tau - 4\tau_{i} -v  \right)^2- \frac{1}{2} \left( \tau -4\tau_i -\eta \right)^2\right] } }}{ |\eta - v|^2  }  \left| \frac{\eta+4\tau_i}{v+4\tau_j} \right|^{2\nu} < + \infty
%\end{gather}
\end{proof}

We define the Hilbert space 
\begin{equation}
H: =L^2\left(\gamma \cup \bigcup_{k=1}^n \frac{1}{\gamma}-4\tau_k ,\mathbb{C}^n\right) \sim  L^2\left(\bigcup_{k=1}^n \frac{1}{\gamma}-4\tau_k, \mathbb{C}^n\right) \otimes L^2(\gamma, \mathbb{C}^n),
\end{equation}
and the matrix operator $\mathbb{K}^{GEN}: H \rightarrow H$ 
\begin{equation}
\mathbb{K}^{GEN} = \left[ \begin{array}{c|c} 
0 & \mathcal{N} \\ \hline
\mathcal{M} & \mathcal{P}
\end{array} \right]
\end{equation}

For now, we denote by $``\det"$ the determinant defined by the Fredholm expansion (\ref{Freddet}); then, $``\det"(\text{Id} - \mathbb{K}_B) = \left. \det\right._2(\text{Id} - \mathbb{K}_B)$, since its kernel is diagonal-free. We also introduce another Hilbert-Schmidt operator
\begin{equation*}
\mathbb{K}^{GEN,2} = \left[ \begin{array}{c|c} 
0 & -\mathcal{N} \\ \hline
0&0
\end{array} \right]
\end{equation*}
whose Carleman determinant ($ \text{det}_2$) is still well defined and $\left. \det \right._2(I - \mathbb{K}^{GEN,2})$ is identically $1$.

We finally perform the following chain of equalities
\begin{gather}
``\det" (\text{Id}_{L^2(\mathbb{R}_+)} - K^{GEN} \chi_{\mathcal{I}}) 
= \left. \det\right._2 \left( \text{Id} - K^{GEN} \chi_{\mathcal{I}}\right)e^{-\operatorname{Tr}(H)}  \nonumber \\
= \left. \det\right._2 \left( \text{Id}_{L^2\left(\bigcup_{k=1}^n \gamma_{-k}\right)} - \mathcal{K}^{GEN} \right)e^{- \operatorname{Tr} (\mathcal{H})}  \nonumber \\
= \left. \det \right._2 (\text{Id}_H - \mathbb{K}^{GEN}) \left. \det \right._2(\text{Id}_H - \mathbb{K}^{GEN,2})
= \left. \det \right._2(\text{Id}_H - \mathbb{K}^{GEN}) \nonumber \\
= `` \det" (\text{Id}_H - \mathbb{K}^{GEN}) 
\end{gather}
The first equality follows from the fact that $K^{GEN} - H$ is diagonal-free; the second equality follows from invariance of the determinant under Fourier transform; the third identity is an application of the following result: given $A$, $B$ Hilbert-Schmidt operators, then
\begin{equation*}
 \left. \det \right._2 (\text{Id} - A) \left. \det \right._2(\text{Id} - B)  = \left. \det \right._2 (\text{Id} - A - B + AB) e^{\operatorname{Tr}(AB)} .
 \end{equation*}
 
It is finally just a matter of computation to show that $\mathbb{K}_B$ is an integrable operator of the form (\ref{IIKSGENBesselmulti1})-(\ref{IIKSGENBesselmulti2}).
\end{proof}

\paragraph{$2\times 2$ case.} As an explanatory example, let's consider a Generalized Bessel process with two times $\tau_1<\tau_2$ and two intervals $I_1:= [0,a]$ and $I_2:=[0,b]$.
\begin{gather}
K^{GEN}(x,y) \cdot \text{diag} \left[ \chi_{[0,a]}(y), \chi_{[0,b]}(y) \right] =  \left(\frac{y}{x}\right)^{\nu} \times \nonumber \\
\left\{  \left[ \begin{array}{cc}
-4\int_{\Sigma} \frac{dt \, ds}{(2\pi i)^2} \frac{e^{ - \frac{1}{2}\left(\tau - \frac{1}{t} \right)^2 + xt + \frac{1}{2}\left(\tau - \frac{1}{s}\right)^2 - ys}}{\left(s-t\right) t^\nu s^{-\nu}} 
&0  \\
-4 \int_{\Sigma} \frac{dt \, ds}{(2\pi i)^2} \frac{e^{ - \frac{1}{2}\left(\tau - \frac{1}{t} \right)^2 + xt + \frac{1}{2}\left(\tau - \frac{1}{s}+ \Delta_{12}\right)^2 - ys}}{\left(s-t+ \Delta_{12} ts\right)t^\nu s^{-\nu}} 
&0
\end{array} \right] \chi_{[0,a]}(y)  \right. \nonumber \\
+ \left.
 \left[ \begin{array}{cc}
0 & -4 \int_{\Sigma} \frac{dt \, ds}{(2\pi i)^2} \frac{e^{ - \frac{1}{2}\left(\tau - \frac{1}{t} \right)^2 + xt + \frac{1}{2}\left(\tau - \frac{1}{s}+ \Delta_{21}\right)^2 - ys}}{\left(s-t+ \Delta_{21} ts\right)t^\nu s^{-\nu}}  \\
0 &-4  \int_{\Sigma} \frac{dt \, ds}{(2\pi i)^2} \frac{e^{ - \frac{1}{2}\left(\tau - \frac{1}{t} \right)^2 + xt + \frac{1}{2}\left(\tau - \frac{1}{s}\right)^2 - ys}}{\left(s-t\right)t^\nu s^{-\nu}}
\end{array} \right] \chi_{[0,b]}(y)
 \right. \nonumber \\
\left. + \left[\begin{array}{cc} 0 & - \frac{1}{\Delta_{21}}  \int_\gamma e^{\frac{x}{4\Delta_{21}}(t-1) + \frac{y}{4\Delta_{21} }\left(  \frac{1}{t}-1  \right)} t^{-\nu -1} \frac{dt}{2\pi i} \\ 0 &0 \end{array} \right]\chi_{[0,b]}(y) \right\}
\end{gather}
with $\Sigma := \gamma \times \hat \gamma$.

Then, the integral operator $\mathbb{K}^{GEN}:H\rightarrow H$ on the space $H:= L^2\left(\gamma \cup \gamma_{-1} \cup \gamma_{-2}, \mathbb{C}^2\right)$  has the following expression
\begin{gather}
\mathbb{K}^{GEN} = \left[ \begin{array}{c|c} 
0 & \mathcal{N} \\ \hline
\mathcal{M} & \mathcal{P}
\end{array} \right] 
\end{gather}
\begin{gather}
\mathcal{N} = \frac{1}{\xi-v} \left[ \begin{array}{cc}
 -4  e^{  \frac{a}{\xi_1}- \frac{a}{v_1} }  \chi_{\gamma_{-1}}(\xi) \chi_{\gamma}(v) & 0 \\
 0 & -4  e^{  \frac{b}{\xi_2} - \frac{b}{v_2} }  \chi_{\gamma_{-2}}(\xi) \chi_{\gamma}(v) \\
\end{array} \right] \\
\mathcal{M} =  \left[ \begin{array}{cc}
 \frac{e^{  \frac{(\xi_{1,\tau})^2}{2}  - \frac{(v_{1,\tau})^2}{2}}}{\xi-v} \frac{ v_1^\nu}{\xi_1^\nu} \chi_{\gamma}(\xi)  \chi_{\gamma_{-1}}(v) 
 &   \frac{e^{  \frac{(\xi_{1,\tau})^2}{2} - \frac{(v_{1,\tau})^2}{2} }}{\xi -v} \frac{v_1^\nu}{\xi_2^\nu} \chi_{\gamma}(\xi)   \chi_{\gamma_{-1}}(v)  \\
\frac{ e^{  \frac{(\xi_{2,\tau})^2}{2} -\frac{(v_{2,\tau})^2}{2} }}{\xi-v}  \frac{v_2^\nu }{\xi_1^\nu} \chi_{\gamma}(\xi)  \chi_{\gamma_{-2}}(v)
 & \frac{e^{  \frac{(\xi_{2,\tau})^2}{2} - \frac{(v_{2,\tau})^2}{2} }}{\xi-v} \frac{v_2^\nu}{\xi_2^\nu} \chi_{\gamma}(\xi)   \chi_{\gamma_{-2}}(v) 
\end{array}\right] \\
\mathcal{P} = \frac{1}{\xi-v} \left[ \begin{array}{cc}
 0 & -4  e^{ \frac{b}{\xi_2}- \frac{b}{v_2 }}  \frac{v_1 ^\nu}{v_2^\nu} \chi_{\gamma_{-2}}(\xi)  \chi_{\gamma_{-1}}(v)  \\
0 & 0 
\end{array}
\right]
\end{gather}
and the equality between Fredholm determinants holds
\begin{equation}
\det\left(\text{Id}_{L^2(\mathbb{R}_+, \mathbb{C}^2)} - K^{GEN} \text{diag}(\chi_{I_1}, \chi_{I_2})\right) = \det \left( \text{Id}_H - \mathbb{K}^{GEN} \right).
\end{equation}

\subsection{Rieman-Hilbert problem and $\tau$-function}
We can now relate the Fredholm determinant of the multi-time Generalized Bessel operator to the isomonodromy  theory by defining a suitable Riemann-Hilbert problem.

\begin{prop} 
The Riemann-Hilbert problem associated to the integrable kernel (\ref{IIKSGENBesselmulti1})-(\ref{IIKSGENBesselmulti2}) is the following:
\begin{align}
& \Gamma_+(\lambda)  = \Gamma_-(\lambda)M(\lambda) \ \ \  \lambda \in \Sigma:= \gamma \cup \left( \bigcup_{j=1}^n \gamma_{-j}\right) \\
& \Gamma(\lambda) = I + \mathcal{O}\left( \frac{1}{\lambda} \right) \ \ \ \lambda \rightarrow \infty \\
& M(\lambda) := I - 2\pi i J^{GEN}(\lambda)
\end{align}
with $\Gamma$ a $(3n-1)\times(3n-1)$ matrix which is analytic on $\mathbb{C} \backslash \Sigma$ and along the collection of curves $\Sigma$ satisfies the above jump condition with
\begin{gather}
J^{GEN} (\lambda) = \textbf{f}(\lambda) \textbf{g}(\lambda)^T =\nonumber \\ 
\left[ \begin{array}{c|c|c}
0 & \operatorname{diag} \mathcal{N}_f(\lambda) \mathcal{M}_g(\lambda) ^T& 0 \\ \hline
\operatorname{diag} \mathcal{M}_f(\lambda) \operatorname{diag} \mathcal{N}_g(\lambda) & 0 & \operatorname{diag}\mathcal{M}_f(\lambda) \mathcal{B}(\lambda)^T \\ \hline
\mathcal{A}(\lambda)^T \operatorname{diag}\mathcal{N}_g(\lambda) & 0 & \mathcal{A}(\lambda)^T \mathcal{B}(\lambda)^T
\end{array} \right]
\end{gather}
\begin{gather}
 \operatorname{diag} \mathcal{N}_f \cdot \mathcal{M}_g^T  = \left[ \begin{array}{ccc}
 -4 e^{-\theta_{1}+\theta_{1,\tau}}\chi_\gamma  & \ldots &-4 e^{-\theta_{1}+\theta_{n,\tau}}  \chi_\gamma \\
 \vdots & & \vdots \\
 -4 e^{ -\theta_{n}+\theta_{1,\tau}} \chi_\gamma & \ldots & -4 e^{-\theta_{n} +\theta_{n,\tau}} \chi_\gamma
 \end{array} \right]  \in \operatorname{Mat}_{n\times n}(\mathbb{C})
 \end{gather}
 \begin{gather}
 \operatorname{diag} \mathcal{M}_f \cdot \operatorname{diag} \mathcal{N}_g =   \operatorname{diag} \left[ e^{\theta_{1}-\theta_{1,\tau}} \chi_{\gamma_{-1}}, \ldots, e^{ \theta_{n}-\theta_{n,\tau}} \chi_{\gamma_{-n}}\right]  \in \operatorname{Mat}_{n\times n}(\mathbb{C})
\end{gather}
\begin{gather}
\operatorname{diag} \mathcal{M}_f \cdot \mathcal{B}^T   \in \operatorname{Mat}_{n\times (n-1)}(\mathbb{C}) \nonumber \\
=  \left[ \begin{array}{cccc}
 0 & && \\
e^{\theta_{2}-\theta_{2,\tau}}\chi_{\gamma_{-2}} & 0  && \\
& e^{\theta_{3}-\theta_{3,\tau}}\chi_{\gamma_{-3}} & \ddots & \\
 &&& e^{\theta_{n}-\theta_{n,\tau}}\chi_{\gamma_{-n}} 
\end{array} \right] 
\end{gather}
\begin{gather}
\mathcal{A}^T \cdot \operatorname{diag}\mathcal{N}_g   \in \operatorname{Mat}_{(n-1)\times n}(\mathbb{C}) \nonumber \\
= \left[ \begin{array}{ccccc}
-4e^{ \theta_{1}-\theta_{2}} \chi_{\gamma_{-1}}& 0 &&& \\
-4e^{\theta_{1} - \theta_{3}} \chi_{\gamma_{-1}}& -4e^{\theta_{2} - \theta_{3}} \chi_{\gamma_{-2}} & 0 && \\
-4e^{ \theta_{1} - \theta_{4}} \chi_{\gamma_{-1}}& -4e^{ \theta_{2}- \theta_{4}} \chi_{\gamma_{-2}} & -4e^{\theta_{3}- \theta_{4} } \chi_{\gamma_{-3}} & 0 &\\
\vdots&&& \ddots& \\
&&&& \\
-4e^{\theta_{1}- \theta_{n}} \chi_{\gamma_{-1}}& \ldots && -4e^{ \theta_{n-1}- \theta_{n}} \chi_{\gamma_{-(n-1)}} & 0 
\end{array} \right] 
\end{gather}
\begin{gather}
\mathcal{A}^T \cdot \mathcal{B}^T   \in \operatorname{Mat}_{(n-1)\times (n-1)}(\mathbb{C}) \nonumber \\
= \left[ \begin{array}{ccccc}
0&&&& \\
-4e^{\theta_{2}-\theta_{3}}\chi_{\gamma_{-2}} &&& & \\
  -4 e^{\theta_{2} -  \theta_{4}} \chi_{\gamma_{-2}}& -4 e^{\theta_{3}- \theta_{4}} \chi_{\gamma_{-3}} && &\\
 \vdots &\vdots &&& \\
&&&& \\
 -4e^{\theta_{2} - \theta_{n}}\chi_{\gamma_{-2}} && -4e^{\theta_{n-1}- \theta_{n}}\chi_{\gamma_{n-1}} & &0
\end{array} \right]
\end{gather} 
with $ \theta_{k}= \frac{a^{k}}{\lambda_k} +\nu \ln \lambda_k$ and $\theta_{h,\tau} =  \frac{(\lambda_{h,\tau})^2}{2}$, $k,h = 1,\ldots, n$.
%with $\theta_{kh} := \frac{a^{k}}{\lambda_k} - \frac{(\lambda_{h,\tau})^2}{2} + \nu \ln \lambda_k$, $k,h= 1, \ldots, n$. 
\end{prop}

%Let's get the jump matrix:
%\begin{gather}
%\textbf{f}(v) = \left[ \begin{array}{c|c}
%\text{diag} \mathcal{N}(v) & 0 \\ \hline
%0 & \text{diag} \mathcal{M} (v) \\ \hline
%0& \mathcal{A}(v)^T
%\end{array}\right] \\ 
%\textbf{g}(\xi)^T = \left[ \begin{array}{c|c|c}
%0 & \mathcal{M}(\xi)^T & 0 \\ \hline
%\text{diag} \mathcal{N} (\xi) & 0 & \mathcal{B}(\xi)^T
%\end{array} \right]
%\end{gather}

\paragraph{$2\times 2$ case.} In the simple $2$-times case, the jump matrix reads
\begin{gather}
J^{GEN}(\lambda) = \nonumber \\
\left[ \begin{array}{ccccc}
0&0& -4e^{-\theta_{1}+\theta_{1,\tau}} \chi_{\gamma}& -4e^{-\theta_{1}+\theta_{2,\tau}} \chi_{\gamma}&0 \\
0&0& -4e^{-\theta_{2}+\theta_{1,\tau}}\chi_{\gamma}& -4e^{-\theta_{2}+\theta_{2,\tau}}\chi_{\gamma}&0 \\
e^{\theta_{1}-\theta_{1,\tau}} \chi_{\gamma_{-1}}&0&0&0&0 \\
0& e^{\theta_{2}-\theta_{2,\tau}}\chi_{\gamma_{-2}}&0&0&  e^{\theta_{2}-\theta_{2,\tau}}\chi_{\gamma_{-2}} \\
\scriptstyle -4e^{\theta_{1} - \theta_{2}} \chi_{\gamma_{-1}}&0&0&0&0
\end{array} \right]
\end{gather}

The jump matrix, though it might look complicated, is equivalent to a matrix with constant entries
\begin{gather}
e^{T^{GEN}} \cdot J_0 \cdot e^{-T^{GEN}} = J^{GEN} \\
T^{GEN} = \text{diag}  \left[-\theta_1, \ldots, -\theta_n, -\theta_{1,\tau}, \ldots, -\theta_{n,\tau}, -\theta_2, \ldots, -\theta_n  \right]
\end{gather}
%with $ \theta_{k}= \frac{a^{k}}{\lambda_k} +\nu \ln \lambda_k$ and $\theta_{h,\tau} =  \frac{(\lambda_{h,\tau})^2}{2}$; so that
so that the matrix $\Psi^{GEN}(\lambda) = \Gamma(\lambda) e^{T^{GEN}(\lambda)}$ solves a Riemann-Hilbert problem with constant jumps and it is a solution to a polynomial ODE.

Referring to the theorems proved in \cite{Misomonodromic} and  \cite{MeM} (see also \cite{MeMmulti}), we can claim
\begin{thm}\label{TEOFREDHOLM}
Given $n$ times $\tau_1 < \tau_2 < \ldots < \tau_n$ and given the interval matrix $\chi_{\mathcal{I}}:= \operatorname{diag} \, \left[ \chi_{I_1}, \ldots, \chi_{I_n}\right] $ with
\begin{equation}
I_j:= \left[a_1^{(j}, a_2^{(j)}\right] \cup \left[a_3^{(j)}, a_4^{(j)} \right] \cup \ldots \cup \left[a_{2k_j-1}^{(j)}, a_{2k_j}^{(j)} \right],
\end{equation}
the Fredholm determinant $\det \left( \operatorname{Id} - K^{GEN}\chi_{\mathcal{I}} \right)$ is equal to the isomonodromic $\tau$-function related to the above Riemann-Hilbert problem. 

In particular, $\forall \, j = 1,\ldots, n $ and $\forall \, \ell = 1, \ldots, 2k_j$ we have
\begin{align}
&\partial \ln \det \left( \operatorname{Id} -  K^{GEN} \chi_{\mathcal{I}} \right) = \int_\Sigma \operatorname{Tr} \left( \Gamma_-^{-1}(\lambda) \Gamma'_-(\lambda)\Xi_\partial (\lambda) \right) \frac{d\lambda}{2\pi i} \\
& \Xi_\partial(\lambda) := \partial M(\lambda) M(\lambda)^{-1} = - 2\pi i \, \partial J^{GEN}  \left( I + 2\pi i J^{GEN} \right)
\end{align}
the $'$ notation means differentiation with respect to $\lambda$, while with $\partial$ we denote any of the partial derivatives $\partial_{\tau_j}$, $\partial_{a_{\ell}^{(j)}}$, $\partial_\tau$.
\end{thm}

\begin{proof}
The following formula holds in general (see \cite{Misomonodromic}) 
\begin{equation}
\omega (\partial) = \partial \ln \det (I-K^{GEN}\chi_{\mathcal{I}}) - H(M) \label{Misomonodromictau}
\end{equation}
where
\begin{gather}
\omega (\partial) := \int_\Sigma \operatorname{Tr}\left( \Gamma_-^{-1}(\lambda) \Gamma_-'(\lambda) \Xi_\partial (\lambda) \right) \, \frac{d\lambda}{2\pi i} \\
H(M) := \int_{\Sigma} \left( \partial \textbf{f}\,  '\, ^T \textbf{g} + \textbf{f}\, '\, ^T \partial \textbf{g} \right) d\lambda - 2\pi i \int_{\Sigma} \textbf{g}^T \textbf{f}\, ' \partial \textbf{g}^T \textbf{f} \, d\lambda. \label{Hannoying}
\end{gather}

Therefore, it is enough to verify that 
%\begin{gather}
%\int_{\Sigma} \left( \partial \textbf{f}\,  '\, ^T \textbf{g} + \textbf{f}\, '\, ^T \partial \textbf{g} \right) d\lambda - 2\pi i \int_{\Sigma} \textbf{g}^T \textbf{f}\, ' \partial \textbf{g}^T \textbf{f} \, d\lambda \nonumber \\
%= \int_{\Sigma} \partial \left( \textbf{f}\,  '\, ^T \textbf{g}\right) d\lambda - 2\pi i \int_{\Sigma} \left( \textbf{f}\, '\, ^T \textbf{g}\right)^T \left( \textbf{f} \,^T  \partial \textbf{g}\right)^T d\lambda  \equiv 0
%\end{gather}
$H(M)\equiv 0$ with $M(\lambda)  = I-J^{GEN}(\lambda)$.
\end{proof}

%\begin{note} 
%The matrix $J^{GEN}$ is a nilpotent matrix so it's true that the inverse of the matrix $I+J^{GEN}$ is just $I-J^{GEN}$.
%\end{note}

%\begin{prop}[stated for general $n$]
%The extra term appearing in \cite{Misomonodromic} is identically zero.
%\begin{gather}
%\int_{\Sigma} \left( \partial \textbf{f}\,  '\, ^T \textbf{g} + \textbf{f}\, '\, ^T \partial \textbf{g} \right) d\lambda - 2\pi i \int_{\Sigma} \textbf{g}^T \textbf{f}\, ' \partial \textbf{g}^T \textbf{f} \, d\lambda \nonumber \\
%= \int_{\Sigma} \partial \left( \textbf{f}\,  '\, ^T \textbf{g}\right) d\lambda - 2\pi i \int_{\Sigma} \left( \textbf{f}\, '\, ^T \textbf{g}\right)^T \left( \textbf{f} \,^T  \partial \textbf{g}\right)^T d\lambda  \equiv 0
%\end{gather}
%\end{prop}
%\begin{proof}
%It follows automatically from noticing that the matrix $\textbf{f}^T(v) \cdot \textbf{g}(\xi)$ consists of entries with products of characterstic functions of non-intersecting curves. Thus, even if we take derivatives, this structure still persists and this imply the resulting zero matrix when $v=\xi=\lambda$ in the integral.
%\end{proof}

Moreover, making use of the Jimbo-Miwa-ueno residue formula (see \cite{Misomonodromic}), it can be shown that
\begin{thm} The following equality holds
%Under our previously mentioned hypotheses (see Theorem \ref{TEOFREDHOLM}), we have
\begin{gather}
\int_\Sigma \operatorname{Tr} \left( \Gamma_-^{-1}(\lambda) \Gamma'_-(\lambda)\Xi_\partial (\lambda) \right) \frac{d\lambda}{2\pi i}  \nonumber \\
= - \operatorname{res}_{\lambda = \infty} \operatorname{Tr}\left( \Gamma^{-1}\Gamma' \partial T^{GEN} \right) + \sum_{i=1}^n \operatorname{res}_{\lambda=-4\tau_i}\operatorname{Tr}\left( \Gamma^{-1}\Gamma' \partial T^{GEN} \right) 
\end{gather}
In particular, regarding the derivative with respect to the endpoints $a^{(j)}$ ($j=1,\ldots, n$)
\begin{gather}
\operatorname{res}_{\lambda=-4\tau_k}\operatorname{Tr}\left( \Gamma^{-1}\Gamma' \partial_{a^{(k)}} T^{GEN} \right) \nonumber \\
=  - \left( \Gamma_{0}^{-1} \Gamma_{1} \right)_{(k,k)} -\chi_{k>1} \left( \Gamma_{0}^{-1} \Gamma_{1} \right)_{(2n-1+k,2n-1+k)} 
\end{gather}
Regarding the derivative with respect to $\tau$
\begin{gather}
\operatorname{res}_{\lambda=\infty} \operatorname{Tr}\left( \Gamma^{-1}\Gamma' \partial_{\tau} T^{GEN} \right)= -\sum_{k=1}^n \Gamma_{1; n+k, n+k}
\end{gather}
Finally, regarding the derivative with respect to the times $\tau_j$ ($j=1,\ldots,n$)
\begin{gather}
 \operatorname{res}_{\lambda=\infty} \operatorname{Tr}\left( \Gamma^{-1}\Gamma' \partial_{\tau_k} T^{GEN} \right) = 4\Gamma_{1;n+k, n+k} \\
 \operatorname{res}_{\lambda=-4\tau_k}\operatorname{Tr}\left( \Gamma^{-1}\Gamma' \partial_{\tau_k} T_B \right)  = -4\nu \left(\Phi_{0;k,k} +  \chi_{k>1}\Phi_{0;2n-1+k, 2n-1+k} \right)  \nonumber \\
+4a^{(k)}(\Phi_{1;k,k} + \chi_{k>1}\Phi_{1;2n-1+k, 2n-1+k} )
\end{gather}
where, given the asymptotic expansion of the matrix $\Gamma \sim \Gamma_0 + \lambda_k \Gamma_1 + \lambda_k^2 \Gamma_2 + \cdots$ in a neighbourhood of $-4\tau_k$, we defined $\Phi_0 := \Gamma_0^{-1}\Gamma_1$ and $\Phi_1 := 2\Gamma_0^{-1}\Gamma_2 - \left(\Gamma_0^{-1}\Gamma_1\right)^2$.
\end{thm}

\begin{oss}
We stated the second part of the theorem above in the simple case $\mathcal{I} = \operatorname{diag} \left[\chi_{[0,a^{(1)}]}, \ldots, \chi_{[0,a^{n}]}\right]$ in order to avoid heavy notation. The general case follows the same guidelines shown in the proof.
\end{oss}

\begin{proof}
We will calculate the residues separately and we will focus on the different parameters ($a^{(j)}$, $\tau_j$ and $\tau$).

\underline{Residue at $\infty$.}
There's no contribution from the residue at infinity when we consider the derivative with respect to the endpoints $a^{(j)}$.
%\begin{gather}
% \Gamma^{-1}\Gamma' \partial_{a^{(j)}} T^{GEN} \sim  \frac{\Gamma_1}{\lambda^2} \cdot  \frac{1}{\lambda_j}   \left( E_{j,j} + \chi_{j> 1} E_{2n-1+j, 2n-1+j}  \right)
%\end{gather}
On the other hand, taking the derivative with respect to the times $\tau_k$ gives:
\begin{gather}
\partial_{\tau_k}T^{GEN} = \nonumber \\
 \left(\frac{4a^{(k)}}{\lambda_k^2}-\frac{4\nu}{\lambda_k} \right)(E_{k,k} + \chi_{k>1}E_{2n-1+k, 2n-1+k} ) - 4\lambda_{k,\tau} E_{n+k, n+k}
\end{gather}
thus the residue is 
\begin{gather}
\operatorname{res}_{\lambda=\infty} \operatorname{Tr}\left( \Gamma^{-1}\Gamma' \partial_{\tau_k} T^{GEN} \right) = 4\Gamma_{1;n+k, n+k} \ \ \ \forall \, k=1, \ldots, n
\end{gather}
We follow a similar argument for the parameter $\tau$:
\begin{gather}
 \partial_{\tau} T^{GEN} = \sum_{k=1}^n \lambda_{k,\tau} E_{n+k, n+k}  
\end{gather}
Thus,
\begin{gather}
\operatorname{res}_{\lambda=\infty} \operatorname{Tr}\left( \Gamma^{-1}\Gamma' \partial_{\tau} T^{GEN} \right)= -\sum_{k=1}^n \Gamma_{1; n+k, n+k}.
\end{gather}

\underline{Residue at $4\tau_k$.}
We recall the asymptotic expansion of the matrix $\Gamma$ in a neighbourhood of $-4\tau_k$:
\begin{gather}
\Gamma \sim \Gamma_0 + \lambda_k \Gamma_1 + \lambda_k^2 \Gamma_2 + \cdots \ \ \ \ \ \ \ \ \lambda \rightarrow -4\tau_k, \ \forall \, k=1,\ldots, n
\end{gather}

\begin{oss} Note that the asymptotic expansion near $-4\tau_k$ is, in general, different for each $k$,
%, i.e. the matrices $\Gamma_j$ are not the same, 
but we wrote them in this way in order to avoid heavy notation.
\end{oss}

%In particular,
%\begin{gather}
%\Gamma^{-1}\Gamma' \sim  \Gamma_0^{-1} \Gamma_1 + \lambda_k \left( -\Gamma_0^{-1}\Gamma_1 \Gamma_0^{-1}\Gamma_1 + 2\Gamma_0^{-1}\Gamma_2\right) + \mathcal{O}\left(\lambda_k^2\right) \nonumber \\
%=: \Phi_{0} + \lambda_k\Phi_{1} + \mathcal{O}\left(\lambda_k^2\right)
%\end{gather}

Regarding the derivative with respect to the endpoints $a^{(k)}$, we have
\begin{gather}
\partial_{a^{(k)}}T^{GEN} = - \frac{1}{\lambda_k}\left[E_{k,k} + \chi_{k>1}E_{2n-1+k, 2n-1+k} \right]
\end{gather}
which implies
\begin{gather}
\text{res}_{\lambda=-4\tau_k}\text{Tr}\left( \Gamma^{-1}\Gamma' \partial_{a^{(k)}} T^{GEN} \right) =  - \left( \Gamma_{0}^{-1} \Gamma_{1} \right)_{(k,k)} -\chi_{k>1} \left( \Gamma_{0}^{-1} \Gamma_{1} \right)_{(2n-1+k,2n-1+k)} 
\end{gather}
and regarding the derivative with respect to the times $\tau_k$, we have
\begin{gather}
\partial_{\tau_k}T^{GEN} = \nonumber \\
\left(\frac{4a^{(k)}}{\lambda_k^2}- \frac{4\nu}{\lambda_k} \right) \left[ E_{k,k} + \chi_{k>1}E_{2n-1+k, 2n-1+k} \right] - 4\lambda_{k,\tau} E_{n+k, n+k}
\end{gather}
thus,
\begin{gather}
\text{res}_{\lambda=-4\tau_k}\text{Tr}\left( \Gamma^{-1}\Gamma' \partial_{\tau_k} T_B \right)  = -4\nu \left(\Phi_{0;k,k} +  \chi_{k>1}\Phi_{0;2n-1+k, 2n-1+k} \right)  \nonumber \\
+4a^{(k)}(\Phi_{1;k,k} + \chi_{k>1}\Phi_{1;2n-1+k, 2n-1+k} )
\end{gather}
There is no contribution from the residue at $-4\tau_k$ ($k=1,\ldots,n$) when taking the derivative with respect to $\tau$.
\end{proof}

\appendix
\section{Building the multi-time Generalized Bessel kernel}

The starting point of this investigation is the known single-time kernel discovered by Kuijlaars \textit{et al.} (\cite{Kuij})
%We have the single time kernel $K^{GEN}_\nu$, 
\begin{gather}
K^{KMW}_\nu(x,y; \tau) 
%= \int_{\hat \gamma} \frac{ds}{2\pi i }\int_{\gamma} \frac{dt}{2\pi i} \, \frac{e^{xs + \frac{\tau}{s} + \frac{1}{2s^2} - yt - \frac{\tau}{t} - \frac{1}{2t^2}}}{t-s} \left(\frac{s}{t} \right)^\nu \nonumber \\
=  \iint_{\gamma \times \hat \gamma} \frac{dt\, ds}{(2\pi i)^2 } \, \frac{e^{-xs - \frac{\tau}{s} + \frac{1}{2s^2} + yt + \frac{\tau}{t} - \frac{1}{2t^2}}}{s-t} \left(\frac{s}{t} \right)^\nu
% \nonumber \\
%= \iint_{\gamma \times \hat \gamma} \frac{dt\, ds}{(2\pi i)^2 } \, \frac{e^{-xs + yt + \frac{1}{2}\left( \tau-\frac{1}{s}  \right)^2 - \frac{1}{2} \left( \tau-\frac{1}{t} \right)^2 }}{s-t} \left(\frac{s}{t} \right)^\nu
\end{gather}
where the curve $\gamma$ is an unbounded curve that extends from $-\infty$ to zero and then back to $-\infty$, encircling the origin in a counterclockwise way, and $\hat \gamma := \frac{1}{\gamma}$; the logarithmic cut is on $\mathbb{R}_-$, as shown in \figurename \  \ref{ArnoGenBes}. 

The diffusion kernel related to the Squared Bessel Paths is
\begin{gather}
p (x,y, \Delta) := \left( \frac{y}{x} \right)^{\frac{\nu}{2}} \frac{1}{\Delta} e^{-\frac{x+y}{4\Delta}} I_\nu\left( \frac{\sqrt{xy}}{2\Delta} \right) 
%\nonumber \\
%= -\frac{1}{\Delta} \int_\gamma e^{\frac{y}{4\Delta}(w-1) + \frac{x}{4\Delta}\left( \frac{1}{w}-1 \right)} w^{-\nu} \frac{dw}{(2\pi i )w}
\end{gather}
where $\Delta>0$ represents the gap between two given times $\tau_i$ and $\tau_j$ and $I_\nu$ is the modified Bessel function of first kind
%; from now on, we will use the notation $\Delta_{ij}:= |\tau_i-\tau_j| >0 $
%, as $\tau_j>\tau_i$
%with $\Delta = \Delta_{ji} := \tau_j-\tau_i >0$ and 
 (the same diffusion kernel appears in the definition of the multi-time Bessel kernel; see \cite{Me}).
 %; in \cite{Kuij} we have $\left(\frac{y}{x}\right)^{\nu/2}$, though).  

%Having a look at \cite{TWPearcey}, the usual shape of a multi-time kernel is the following:
The extended multi-time kernel is given by
\begin{gather}
K = H- P_\Delta
\end{gather}
where in particular $P_\Delta$ is a strictly upper-triangular matrix with $(i,j)$-entry $P_{\Delta_{ij}}:= p(x,y, \Delta_{ij})$ when $j>i$ ($\Delta_{ij}:= |\tau_i-\tau_j| >0 $). This is essentially the derivation in \cite{Eynard} applied to case at hand. 
\begin{thm}
The multi-time Generalized Bessel operator on $L^2(\mathbb{R}_+)$ with times $\tau_1<\ldots< \tau_n$ is defined through a matrix kernel with the following entries $K^{GEN}:= H_{ij} + \chi_{i<j} P_{\Delta_{ij}}$ ($i,j=1,\ldots, n$)
\begin{gather}
H_{ij}(x,y) :=  4 \iint_{\hat \gamma \times \gamma} \frac{ds \, dt}{(2\pi i)^2}  \frac{e^{  - xs + yt + \frac{1}{2}\left( \tau-\frac{1}{s} + 4\Delta_{ji}  \right)^2 - \frac{1}{2} \left( \tau-\frac{1}{t} \right)^2 }}{ \left(t - s - 4\Delta_{ji}ts \right)} \left( \frac{s}{t} \right)^\nu \label{GBcopyright1}
\\
P_{\Delta_{ij}} (x,y) 
%:= \left( \frac{y}{x} \right)^{\frac{\nu}{2}} \frac{1}{\Delta_{ji}} e^{-\frac{x+y}{4\Delta_{ji}}} I_\nu\left( \frac{\sqrt{xy}}{2\Delta_{ji}} \right)  \nonumber \\
= - \frac{1}{\Delta_{ji}} \int_\gamma e^{ \frac{x}{4\Delta_{ji}}\left( \frac{1}{w}-1 \right)+\frac{y}{4\Delta_{ji}}(w-1)} w^{-\nu} \frac{dw}{(2\pi i )w}  \label{GBcopyright2}
\end{gather}
the curve $\gamma$ is the same one as in the single-time Extended Bessel kernel (a contour that winds around zero counterclockwise an extends to $-\infty$) and $\hat\gamma := \frac{1}{\gamma}$; $\Delta_{ji}:= \tau_j-\tau_i>0$.
\end{thm}

The proof is based on the verification that the definition of the kernel above satisfies the Eynard-Mehta theorem on multi-time kernels (\cite{Eynard}). 

First of all, we define a convolution operation (see \cite[formula (3.2)]{Eynard}).
\begin{defn}
Given two functions $f,g$ with suitable regularity, we define the convolution $f\ast g$ as
\begin{equation}
(f\ast g) (\xi,\eta) = \int f(\xi,\zeta) g(\zeta,\eta) \, d\zeta.
\end{equation}
\end{defn}

Recalling formul\ae \ (3.12)-(3.13) from \cite{Eynard}, we will verify the following relations between the diffusion kernel $P_\Delta$ and the kernel $H$:
%Following the argument in the same paper, we get to the following conclusion (see )
\begin{gather}
H_{ij} \ast P_{\Delta_{jk}} = \left\{ \begin{array}{cc} H_{ik} & j<k \\ 0 & j\geq k \end{array} \right. \\
P_{\Delta_{ij}}\ast H_{jk}= \left\{ \begin{array}{cc} H_{ik} & i<j \\ 0 & i\geq j \end{array} \right. \label{convKP22}
\end{gather}

\begin{proof} 
We set $\Delta := |\tau_i-\tau_j| >0$. Regarding the upper diagonal terms ($i<j$)
\begin{gather}
H_{ij}(x,y)=\int_0^\infty P_{\Delta_{ij}}(x,z) H_{jj}(z,y)  \, dz = \nonumber \\
  \int_0^\infty \frac{dz}{ \Delta}  \int_\gamma \frac{dw}{2\pi iw} e^{\frac{x}{4\Delta}\left( \frac{1}{w}-1 \right) +\frac{z}{4\Delta}(w-1)} w^{-\nu} \,  \iint_{\gamma \times \hat \gamma} \frac{dt\, ds}{(2\pi i)^2} \frac{e^{-zs + yt + \frac{1}{2}\left( \tau-\frac{1}{s}  \right)^2 - \frac{1}{2} \left( \tau-\frac{1}{t} \right)^2 }}{t-s} \left(\frac{s}{t} \right)^\nu
\end{gather}

Integrating in $z$ and taking calculating a residue, we have
\begin{gather}
%= \frac{1}{\Delta}   \iint_{\gamma \times \hat \gamma_u} \frac{dt\, du}{(2\pi i)^2} \frac{e^{-xu+ yt + \frac{1}{2}\left( \tau+4\Delta-\frac{1}{u}  \right)^2 - \frac{1}{2} \left( \tau-\frac{1}{t} \right)^2 }}{\left(\frac{u}{1-4\Delta u}-t\right)\left( 4\Delta\frac{ u}{1-4\Delta u} +1 \right)} \left(\frac{u}{t} \right)^\nu \frac{1}{(1-4\Delta u)^2} \nonumber \\
 \frac{1}{\Delta}   \iint_{\gamma \times \hat \gamma_u} \frac{dt\, du}{(2\pi i)^2} \frac{e^{-xu+ yt + \frac{1}{2}\left( \tau+4\Delta-\frac{1}{u}  \right)^2 - \frac{1}{2} \left( \tau-\frac{1}{t} \right)^2 }}{\left(u-t +4\Delta u t\right)} \left(\frac{u}{t} \right)^\nu 
\end{gather}
%{\color{red} on the other hand the result is clearly analytic on $[\alpha, +\infty)$ for every $\alpha>0$, thus the extra cut does not subsist. (Not sure about this, but it's nice, especially in the part where I want to recover the single-time kernel)}

As for the lower diagonal term, we need to verify that
\begin{gather}
\int_0^\infty P_{\Delta_{ij}}(x,z) H_{ji}(z,y) \, dz = H_{ii}(x,y)  \ \ \ j>i
\end{gather}
with
\begin{gather}
H_{ji}(x,y) =  \frac{4}{(2\pi i)^2} \iint_{\hat \gamma \times \gamma} \frac{du \, dt}{ut}  \frac{e^{  - xu + yt + \frac{1}{2}\left( \tau-\frac{1}{u} - 4\Delta_{ij}  \right)^2 - \frac{1}{2} \left( \tau-\frac{1}{t} \right)^2 }}{ \left(\frac{1}{u} - \frac{1}{t} + 4\Delta_{ij} \right)} \left( \frac{u}{t} \right)^\nu
\end{gather}

Again ,we set $\Delta := |\tau_i-\tau_j| >0$. 
\begin{gather}
 \int_0^\infty \frac{4\, dz}{\Delta} \int_\gamma \frac{dw}{(2\pi i )w} e^{\frac{x}{4\Delta}\left( \frac{1}{w}-1 \right) +\frac{z}{4\Delta}(w-1) } w^{-\nu}    \iint_{\hat \gamma \times \gamma} \frac{du \, dt}{(2\pi i)^2}  \frac{e^{ -zu +yt + \frac{1}{2}\left( \tau-\frac{1}{u} -4\Delta  \right)^2 - \frac{1}{2} \left( \tau-\frac{1}{t} \right)^2 }}{ \left(u - t - 4\Delta tu \right)} \left( \frac{u}{t} \right)^\nu
\end{gather}
we integrate in $z$ and calculate a residue to get
% with respect to $w = 4\Delta u +1$, gives
%Integrate in $z$
%\begin{gather}
%= \frac{4}{\Delta} \int_\gamma \frac{dw}{(2\pi i)w}   \iint_{\hat \gamma_u \times \gamma} \frac{du \, dt}{(2\pi i)^2 tu}  \frac{e^{\frac{x}{4\Delta}\left( \frac{1}{w}-1 \right) +yt + \frac{1}{2}\left( \tau-\frac{1}{u} -4\Delta  \right)^2 - \frac{1}{2} \left( \tau-\frac{1}{t} \right)^2 }}{ \left(\frac{w-1}{4\Delta }-u \right) \left(\frac{1}{u} - \frac{1}{t} + 4\Delta  \right)} \left( \frac{u}{tw} \right)^\nu 
%\end{gather}
%Now again the integrand converges in $w$ at $\infty$ and the $w$-curve $\gamma$ is always on the left of $4\Delta \hat \gamma_u +1$. We take the residue again at $w=4\Delta u +1$:
%\begin{gather}
%= \frac{4}{\Delta}   \iint_{\hat \gamma \times \gamma} \frac{du \, dt}{(2\pi i)^2 }  \frac{e^{- x\frac{u}{4\Delta u +1} +yt + \frac{1}{2}\left( \tau-\frac{1}{u} -4\Delta  \right)^2 - \frac{1}{2} \left( \tau-\frac{1}{t} \right)^2 }}{ tu  \left(\frac{1}{u} - \frac{1}{t} + 4\Delta  \right) (4\Delta u +1)} \left( \frac{u}{t(4\Delta u +1)} \right)^\nu 
%\end{gather}
%Finally we apply the change of variable $u = \frac{v}{1-4\Delta v}$:
%\begin{gather}
%u = \frac{v}{1-4\Delta v} \ \Leftrightarrow \  \frac{1}{u} = \frac{1}{v} - 4\Delta \ \Leftrightarrow \ v = \frac{u}{4\Delta u + 1} \\
%du = \frac{dv}{(1-4\Delta v)^2}
%\end{gather}
%with the $v$-curve looking again like $\hat \gamma$ with an extra cut on $[\frac{1}{8\Delta}, +\infty)$; so
\begin{gather}
%= \frac{4}{\Delta}   \iint_{\hat \gamma_v \times \gamma} \frac{dv \, dt}{(2\pi i)^2 }  \frac{e^{- xv +yt + \frac{1}{2}\left( \tau-\frac{1}{v} + 4\Delta -4\Delta  \right)^2 - \frac{1}{2} \left( \tau-\frac{1}{t} \right)^2 }}{ t \frac{v}{1-4\Delta v}  \left(\frac{1}{v} - 4\Delta - \frac{1}{t} + 4\Delta  \right) \left(4\Delta \frac{v}{1-4\Delta v} +1\right)} \left( \frac{v}{t} \right)^\nu  \frac{1}{(1-4\Delta v)^2} \nonumber \\
  \frac{4}{\Delta}   \iint_{\hat \gamma \times \gamma} \frac{dv \, dt}{(2\pi i)^2 }  \frac{e^{- xv +yt + \frac{1}{2}\left( \tau-\frac{1}{v}  \right)^2 - \frac{1}{2} \left( \tau-\frac{1}{t} \right)^2 }}{ \left(t-v  \right) } \left( \frac{v}{t} \right)^\nu  
\end{gather}
\end{proof}

Independently from the present work and almost simultaneously, Veto and Delvaux (\cite{Balint}) introduced another version of the multi-time Generalized Bessel operator, called Hard-edge Pearcey process. 

The kernel of the Hard-edge Pearcey reads $L^{\nu}:= H - P_\Delta$ with entries
\begin{gather}
H_{ij}(x,y, \sigma):= \nonumber\\
\left(  \frac{y}{x}\right)^\nu \int_{\Gamma_{-\tau_i}} \frac{d\eta}{2\pi i} \int_{i\mathbb{R}+\delta} \frac{d\xi}{2\pi i} \frac{e^{ - \frac{1}{2}\left(\eta-\sigma \right)^2 + \frac{x}{\eta+\tau_i} - \frac{1}{2}\left(\xi-\sigma\right)^2 - \frac{y}{\xi+\tau_j}}}{(\eta-\xi)(\eta+\tau_i)(\xi+\tau_j)} \left(\frac{\eta+\tau_i}{\xi+\tau_j}\right)^\nu 
\end{gather}
and $P_{\Delta}$ the usual transition density as above. $\Gamma_{-\tau_i}$ is a clockwise oriented closed loop which intersects the real line at a point to the right of $-\tau_i$, and also at $-\tau_i$ itself, where it has a cusp at angle $\pi$; $\delta >0$ is chosen such that the contour $i\mathbb{R} + \delta$ passes to the right of the singularity at $-t$ and to the right of the contour $\Gamma_{-\tau_i}$. The logarithmic branch is cut along the negative half-line.

\begin{prop}
The Hard-edge Pearcey operator is the transposed of the Generalized Bessel operator (\ref{GBcopyright1})-(\ref{GBcopyright2}) defined above. More precisely,
\begin{equation}
L^\nu (x,y; \sigma) = \left(\frac{y}{x}\right)^\nu K^{GEN} (y, x ;  \tau_i+\sigma)
\end{equation}
\end{prop}

\begin{proof}
The results comes from straightforward changes of variables.
\end{proof}

%Now, we translate both curves around zero:
%\begin{gather}
%u = \eta + \tau_i \\
%w = \xi+ \tau_j
%\end{gather}
%\begin{gather}
%= \left(  \frac{y}{x}\right)^\nu \int_{\hat \gamma} \frac{du}{2\pi i} \int_{i\mathbb{R}+\delta \sim \gamma} \frac{dw}{2\pi i} \frac{e^{ - \frac{1}{2}\left(u -\tau_i - \sigma \right)^2 + \frac{x}{u} - \frac{1}{2}\left(w-\tau_j -\sigma \right)^2 - \frac{y}{w}}}{uw(u-\tau_i-w+\tau_j)} \left(\frac{u}{w}\right)^\nu 
%\end{gather}

%Last change of variables:
%\begin{gather}
%u \rightarrow \frac{1}{u}=: t, \ \ \ w \rightarrow \frac{1}{w}=: s
%\end{gather}
%\begin{gather}
%= \left(  \frac{y}{x}\right)^\nu \int_{\gamma} \frac{dt}{2\pi i} \int_{ \hat \gamma} \frac{ds}{2\pi i} \frac{e^{ - \frac{1}{2}\left(\frac{1}{t} -\tau_i - \sigma \right)^2 + xt - \frac{1}{2}\left(\frac{1}{s}-\tau_j -\sigma \right)^2 - ys}}{st\left(\frac{1}{t}-\tau_i-\frac{1}{s}+\tau_j\right)} \left(\frac{s}{t}\right)^\nu  \nonumber \\
%= \left(  \frac{y}{x}\right)^\nu \int_{\gamma} \frac{dt}{2\pi i} \int_{ \hat \gamma} \frac{ds}{2\pi i} \frac{e^{ - \frac{1}{2}\left(\tau - \frac{1}{t} \right)^2 + xt - \frac{1}{2}\left(\tau - \frac{1}{s}+ \Delta_{ji}\right)^2 - ys}}{\left(s-t+ \Delta_{ji} ts\right)} \left(\frac{s}{t}\right)^\nu
%\end{gather}

%Here is the one-to-one correspondence with my notation:
%\begin{align} 
%\alpha \ \longleftrightarrow & \ \nu \\ 
%s \ \longleftrightarrow & \ \tau_i \\ 
%t  \ \longleftrightarrow & \ \tau_j  \\
%\sigma \ \longleftrightarrow & \ \tau - \tau_i
%\end{align}
%in particular, the last one comes from Corollary 1.7.

\begin{cor}[\cite{Balint}]
In the single-time case ($\tau_i=\tau_j$), both the Generalized Bessel kernel and the Hard-edge Pearcey kernel coincide with the single-time kernel defined in \cite{Kuij}, up to a transposition:
\begin{equation}
\left.L^\nu (x,y;\sigma)\right|_{ \Delta_{ij}=0} = \left(\frac{y}{x}\right)^\nu K^{KMW}_\nu (y, x;\tau_i+\sigma) = \left(\frac{y}{x}\right)^\nu \left. K^{GEN}(y,x;\tau_i+\sigma)\right|_{ \Delta_{ij}=0}
\end{equation}
\end{cor}

\begin{oss}
Since the Fredholm determinant is invariant under transposition, we preferred to work on the version given by Veto and Delvaux for the multi-time Generalized Bessel operator, because of more straighforward calculations which reminds more closely the ones performed for the Bessel process (\cite{Me}).
\end{oss}

\section*{Acknoledgements}
The author would like to thank Dr. B\'alint Vet\"o for the useful exchange of emails and for being willing to share his and Dr. Delvaux's results.

\bibstyle{plain}
\bibliography{GenBessel}

\end{document}